%% file: main.tex
\documentclass[11pt,a4paper]{article}
\usepackage[a4paper, portrait, margin=2.5cm]{geometry}
\usepackage{amsmath,amssymb,amsthm}
\usepackage{graphicx}
\graphicspath{{./}}
\usepackage{booktabs}
\usepackage[colorlinks=true, citecolor=blue, linkcolor=blue,
            urlcolor=blue]{hyperref}
\usepackage{authblk}
\usepackage{mathpazo}
\usepackage{microtype}
\usepackage{float}
\usepackage{bm}
\usepackage{siunitx}

\newtheorem{theorem}{Theorem}

\newtheorem{corollary}[theorem]{Corollary}
\newtheorem{proposition}[theorem]{Proposition}
\theoremstyle{remark}
\newtheorem{remark}{Remark}

\newcounter{manualthm}
\begin{document}
\input{dynamic_variables.tex}


\title{\LARGE\bfseries The Dynamic Origin of Kleiber's Law\\[4pt]
and the Generalized Metabolic Scaling Theorem}
\author{Riccardo Marchesi}
\affil{University of Pavia}
\date{\today}

\maketitle

\begin{abstract}
The ubiquitous $3/4$ metabolic scaling exponent, known as Kleiber's law,
has long been attributed to the minimization of viscous dissipation within
fractal transport networks. In this paper, we invert this standard
narrative, demonstrating that Kleiber's law is fundamentally a signature
of pulsatile wave physics rather than steady-state geometry. By coupling
local branching optimization to global allometry, we derive the exact
generalized metabolic exponent $\beta = d\alpha/(2d+\alpha)$, which
strictly maps local transport microphysics to global organismal scaling.
We show that dynamic wave-impedance matching in the proximal vasculature
uniquely enforces $\beta = 3/4$ in three dimensions. This bound is
dynamically protected: no static optimization of a viscous network can
reproduce it. Consequently, we analytically predict the critical body
mass for the wave-to-viscous transition, successfully explaining the
empirical shift to steeper allometric scaling ($\beta \approx 0.9$) in
small mammals and invertebrates with no free parameters. Furthermore, we
demonstrate that the classical West--Brown--Enquist (WBE) derivation is
structurally divergent under its own geometric assumptions, failing at
the required proximal-dominance limit. Our framework is validated across
nine biological systems spanning five phyla---including vertebrate
vasculature, insect tracheae, plant xylem, and sponge canals---accurately
predicting empirical branching exponents from independent biophysical
measurements. Ultimately, we establish a general allometric equation of
state that organizes diverse biological networks into discrete universality
classes, generating falsifiable predictions across clades from shrews to
flatworms.
\end{abstract}

\section{Introduction}

Murray's cubic law $\alpha = 3$ has been the cornerstone of biological
transport theory since 1926~\cite{murray1926}. While its
derivation---balancing Poiseuille dissipation against volume
maintenance---is elegant, it has long been recognized as a special case
of a more general efficiency principle. This principle can be formalized
precisely: for a cost function $\Phi(r,X) = Ar^{-n} + Br^m$, the
cost-minimizing branching exponent is universally $\alpha_t = (n+m)/2$
(Theorem~\ref{thm:branching} below; established in full generality in Paper~II~\cite{paperII}; see also~\cite{bennett2025}).

In this paper we pursue two complementary goals. First, we demonstrate
that $\alpha_t = (n+m)/2$ has genuine predictive power when $n$ and $m$
are interpreted as derived biophysical quantities---measurable
independently from histological, biochemical, and biomechanical
data---rather than as phenomenological parameters. We validate this
across nine biological systems spanning vascular, pulmonary, neural,
tracheal, plant, and invertebrate networks, recovering empirical
exponents with no fitted parameters.

Second, and more unexpectedly, we show that the local branching law is the
microscopic foundation for global allometric scaling. Previous models, most
notably the West--Brown--Enquist (WBE) framework~\cite{west1997}, derived
Kleiber's $3/4$ metabolic scaling by assuming $\alpha = 3$. We show instead that
the global metabolic exponent is an exact function of the local geometry:
$\beta(\alpha, d) = \frac{d\alpha}{2d+\alpha}$. This relation reveals that
Kleiber's law is not a consequence of Murray's law---which yields isometric
scaling $\beta = 1$---but rather emerges from the wave-impedance matching
attractor $\alpha_w = 2$ governing the proximal vasculature of pulsatile
cardiovascular systems. This constitutes a fundamental inversion of the standard
allometric narrative, and provides a mechanistic explanation for the empirically
observed deviations from $3/4$ scaling in small mammals and invertebrates.

\section{The General Two-Term Branching Problem}

Consider a branching conduit carrying a conserved scalar flux $X$ (volume flow,
current, or diffusive flux) that bifurcates into two daughter branches. The
total cost per unit length at each vessel of radius $r$ carrying flux $X$ is
assumed to take the two-term power-law form:
\begin{equation}
\Phi(r, X) = A(X)\, r^{-n} + B\, r^{m}, \qquad n, m > 0,
\end{equation}
where the first term represents transport dissipation and the second represents
structural maintenance. The condition $A(X) \propto X^2$ holds in the linear
transport regime---Poiseuille flow, Ohmic conduction, and Fickian diffusion all
satisfy this exactly, since dissipated power scales as flux squared divided by
conductance.

\vspace{0.5em}
\noindent \textbf{Theorem 0 (Optimal Branching Exponent).}\refstepcounter{manualthm}\label{thm:branching} \textit{Let a bifurcation conserve flux ($X_0 = X_1 + X_2$) and minimize total cost $\mathcal{C} = \sum_i \ell_i \Phi(r_i, X_i)$ over radii $\{r_i\}$ at fixed fluxes $\{X_i\}$. Then:}

\textit{(i) The cost-minimizing radius satisfies $r_i^* \propto X_i^{2/(n+m)}$ uniquely.}

\textit{(ii) Under flux conservation, the local network geometry obeys the branching rule:}
\begin{equation}
r_0^{\alpha_t} = r_1^{\alpha_t} + r_2^{\alpha_t}, \qquad \alpha_t =
\frac{n+m}{2}.
\end{equation}

\textit{(iii) The exponent $\alpha_t$ is the \textbf{unique} solution: no other branching rule simultaneously minimizes cost at every node under flux conservation.}

\vspace{0.5em}
\noindent\textbf{Proof.}
\textit{(i)} At fixed $X_i$, minimize $\ell_i \Phi(r_i, X_i)$ over $r_i$:
\begin{equation*}
\frac{\partial \Phi}{\partial r_i} = -n A(X_i) r_i^{-n-1} + m B r_i^{m-1} = 0 \implies r_i^* = \left(\frac{n A(X_i)}{m B}\right)^{1/(n+m)}.
\end{equation*}
Since $A(X_i) \propto X_i^2$, we obtain $r_i^* = c\, X_i^{2/(n+m)}$, uniquely
(the second derivative $\partial^2\Phi/\partial r_i^2 > 0$ confirms a global
minimum).

\textit{(ii)} Inverting the optimal scaling law yields $X_i \propto (r_i^*)^{(n+m)/2}$. Applying the flux conservation constraint $X_0 = X_1 + X_2$ immediately gives:
\begin{equation*}
(r_0^*)^{(n+m)/2} = (r_1^*)^{(n+m)/2} + (r_2^*)^{(n+m)/2},
\end{equation*}
which is Murray's form $r_0^\alpha = r_1^\alpha + r_2^\alpha$ with $\alpha_t =
(n+m)/2$.

\textit{(iii) Uniqueness.} Suppose some other exponent $\alpha \neq \alpha_t$ satisfies the branching rule for all admissible flux ratios $X_1/X_0 \in (0,1)$. The rule $r_0^\alpha = r_1^\alpha + r_2^\alpha$ combined with $r_i \propto X_i^{1/\alpha_t}$ requires:
\begin{equation*}
X_0^{\alpha/\alpha_t} = X_1^{\alpha/\alpha_t} + X_2^{\alpha/\alpha_t}
\end{equation*}
for all $X_1 + X_2 = X_0$. This holds identically as a functional equation if
and only if $\alpha/\alpha_t = 1$, i.e., $\alpha = \alpha_t = (n+m)/2$.
$\square$

\vspace{0.5em}
\noindent\textbf{Remarks on scope.} The result is independent of: branching number $N$ (bifurcations, trifurcations, etc.); branch lengths $\ell_i$; and the specific form of $A(X)$ beyond the linear-regime condition $A(X) \propto X^2$. It applies to any network where transport physics is linear and costs decompose additively across branches. Deviations arise in turbulent flow ($A(X) \propto X^3$, non-linear regime) or three-term cost functions ($m \neq k$); the latter are now fully characterised by Corollary~7 of Paper~I~\cite{paperI}, which establishes that any two distinct maintenance exponents produce non-universal $\alpha$. The generalization to non-symmetric bifurcations and uniqueness under weaker regularity conditions is treated in~\cite{bennett2025}.

\vspace{0.5em}
Three immediate corollaries follow:

\vspace{0.5em}
\noindent \textbf{Corollary 1 (Murray's law).} \textit{Poiseuille flow ($n=4$), volume maintenance ($m=2$): $\alpha_t = 3$.}

\vspace{0.5em}
\noindent \textbf{Corollary 2 (Wall-dominated limit).} \textit{Poiseuille flow ($n=4$), wall-tissue scaling $h \propto r^p$ ($m=1+p$): $\alpha_t = (5+p)/2 \in (2.5, 3)$ for $p \in (0,1)$.}

\vspace{0.5em}
\noindent \textbf{Corollary 3 (Dendritic minimax gap).} \textit{Cytoplasmic flow ($n=4$), surface maintenance ($m=1$): $\alpha_t = 5/2$. The wave attractor $\alpha_w = 2 < \alpha_t$ opens a minimax gap yielding $\alpha^* \approx \alphastarNeural{}$.}
\vspace{0.5em}

The present paper is dedicated to the physical grounding and biological
validation of this theorem---establishing that $n$ and $m$ are not free
parameters, but measurable biophysical realities---and to its extension to
global metabolic scaling (Section 5).

\begin{figure}[!h]
\centering
\includegraphics[width=0.85\textwidth]{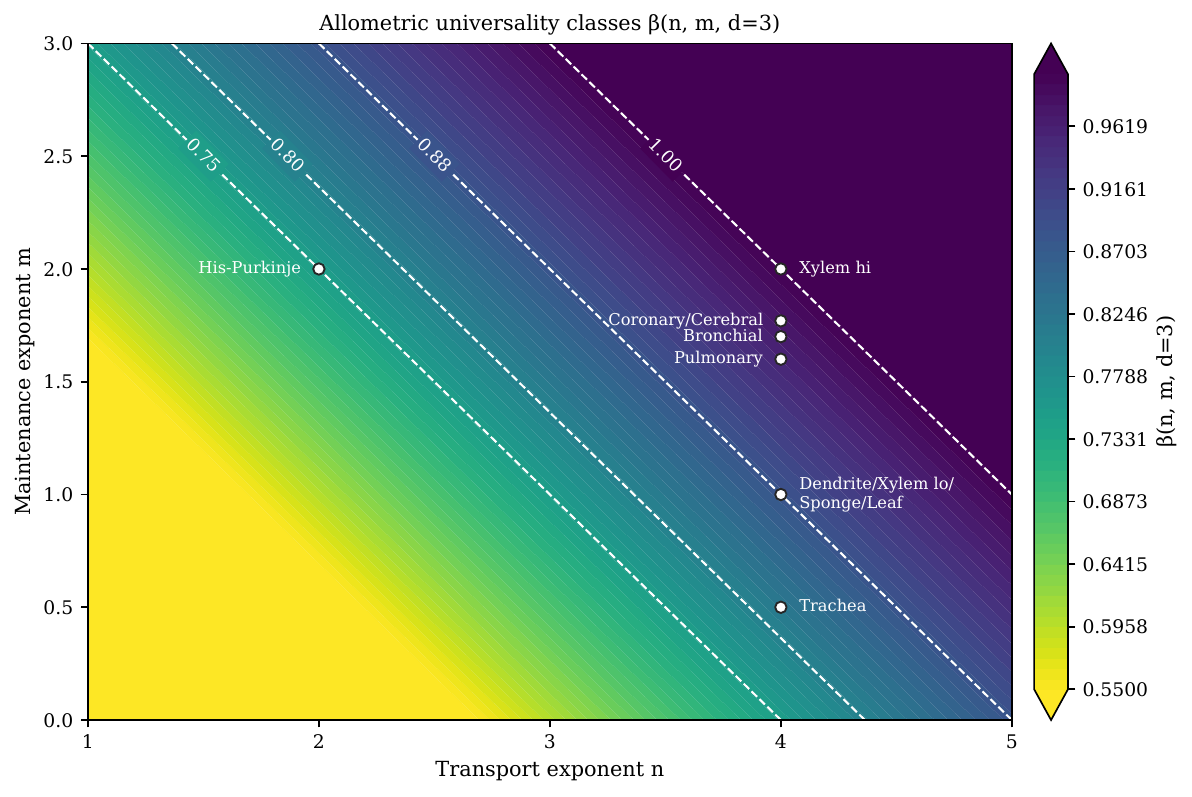}
\caption{Allometric universality classes: the metabolic scaling exponent $\beta(n,m,d{=}3) = 3(n+m)/(12+n+m)$ as a function of the transport exponent $n$ and the maintenance exponent $m$ in three spatial dimensions. Contour lines mark $\beta = 3/4$ (Kleiber), $0.80$, $15/17$ (viscous sub-bound), and $1$ (Murray limit). White circles show the nine biological systems of Table~1. The discrete nature of biologically realizable $(n, m)$ values pins each system to a specific universality class; no continuous tuning between classes is possible.}
\label{fig:spectrum}
\end{figure}

\begin{table}[!h]
\centering
\small
\setlength{\tabcolsep}{5pt}
\caption{Evaluation of $\alpha_t = (n+m)/2$ for nine biological systems spanning five phyla. Parameters $n$ and $m$ are completely determined a priori by microphysical and histological properties independent of the macroscopic branching geometry cited in the $\alpha_{\exp}$ column. $^\dagger$In sponges, $\alpha_{\exp} \approx 2$ is observed, but the responsible mechanism is distinct from wave-impedance matching: $\mathrm{Wo} \ll 1$ structurally precludes the wave attractor, and the case remains an open problem; see Section~\ref{sec:sponge_open}.}
\vspace{0.2cm}
\begin{tabular}{@{}l c c c p{3.2cm} c c@{}}
\toprule
\textbf{System} & $\boldsymbol{n}$ & $\boldsymbol{m}$ & $\boldsymbol{\alpha_t}$ & \textbf{Source of $\boldsymbol{m}$} & $\boldsymbol{\alpha_{\exp}}$ & \textbf{Ref.} \\
\midrule
Coronary artery (porcine) & 4 & \mCoronary{} & \atCoronary{} & Rhodin 1967,
$p=0.77$ & $2.70\pm0.20$ & \cite{kassab1993} \\
Pulmonary artery (human)  & 4 & \mPulmonary{} & \atPulmonary{} & Huang 1996,
$p=0.60$  & 2.65          & \cite{huang1996} \\
Cerebral arteries (human) & 4 & \mCoronary{} & \atCoronary{} & Rhodin 1967,
$p=0.77$ & $2.9\pm0.7$   & \cite{rossitti1993}\\
Bronchial tree            & 4 & \mBronchial{} & \atBronchial{} & Weibel 1963,
$p\approx0.70$ & 2.85    & \cite{weibel1963}\\
Neuronal dendrites        & 4 & \mNeuron{} & \atNeuron{} & Attwell \& Laughlin 2001 &
2.0--2.5 & \cite{cuntz2010}\\
Insect tracheae (trunks)  & 4 & \mTrachea{} & \atTrachea{} & Westneat 2003$^\mathparagraph$ &
2.25          & \cite{westneat2003}\\
Plant xylem               & 4 & \mXylemLo{}--\mXylemHi{} & \atXylemLo{}--\atXylemHi{} &
McCulloh 2003, $p\in[0,1]$ & 2.20--3.00 & \cite{mcculloh2003}\\
Leaf venation ($d=2$)     & 4 & \mLeaf{} & \atLeaf{} & McCulloh 2003          &
2.0--2.5      & \cite{mcculloh2003}\\
Sponge aquiferous canals$^\dagger$ & 4 & \mSponge{} & \atSponge{} & Reiswig 1975 &
$\approx 2.0$ & \cite{reiswig1975}\\
His-Purkinje (mammals) & 2 & \mHP{} & \atHP{} & Weidmann 1952~\cite{weidmann1952}, Ono et al.\ 2009~\cite{ono2009}
& --- (locked) & \cite{noujaim2004}$^{\ddagger}$ \\
\bottomrule
\end{tabular}
\vspace{0.15cm}
\raggedright \footnotesize $^{\ddagger}$ Kinematic signature only: $\mathrm{PR} \propto M^{1/4}$ is an indirect temporal corroboration, not a direct calorimetric measurement of $\beta$.

\raggedright \footnotesize $^\mathparagraph$ The sub-linear wall scaling $h \propto r^{1/2}$ ($m = 0.5$) reflects the taenidia spacing law in insect tracheae (Westneat 2003); interspecific variation in taenidia geometry ($m \in [0.4, 0.6]$ across Diptera and Orthoptera) shifts the predicted $\alpha_t$ by $\pm 0.05$, within the empirical scatter.
\end{table}

\section{System-by-System Derivation of (n, m)}

\subsection{Vascular systems (coronary, pulmonary, cerebral)}
\textbf{Transport term:} Poiseuille dissipation, $P_{\text{visc}} = 8\mu L Q^2 / (\pi r^4)$, giving $n=4$, $A(Q) \propto Q^2$.

\noindent\textbf{Maintenance term:} vessel-wall tissue with $h(r) = c_0 r^p$, so wall volume $\propto r \cdot h \cdot L \propto r^{1+p} \cdot L$, giving $m = 1+p$.
\begin{itemize}
    \item \textit{Coronary/cerebral:} $p = 0.77$ (Rhodin 1967, across mammalian species) $\implies m = \mCoronary{}$, $\alpha_t = \atCoronary{}$.
    \item \textit{Pulmonary:} $p = 0.60$ (Huang et al.\ 1996) $\implies m = \mPulmonary{}$, $\alpha_t = \atPulmonary{}$.
\end{itemize}

\noindent\textbf{Pulmonary minimax gap.} The main pulmonary artery operates in the
strongly pulsatile regime: with $r_0 \approx \rZeroPulm{}\,\mathrm{mm}$ and heart
rate $\omega \approx \omegaPulm{}\,\mathrm{rad\,s^{-1}}$, the Womersley number is
$\mathrm{Wo} \approx \WoPulm{}$, comparable to the aorta. The wave-impedance
attractor $\alpha_w = 2$ therefore competes with the transport optimum
$\alpha_t = \atPulmonary{}$ ($m = \mPulmonary{}$), opening a minimax gap. Solving the
equal-cost condition $f(\alpha^*) = g(\alpha^*)$ for $(n=4,\, m=\mPulmonary{},\,
K=11,\, N=2)$ yields:
\begin{equation}
\alpha^*_{\mathrm{pulm}} \approx \alphastarPulm{}, \qquad
f(\alpha^*) = g(\alpha^*) \approx \fstarPulm{}\%.
\end{equation}
This is consistent with the morphometric value $\alpha_{\exp} = 2.65$
reported by Huang et al.~\cite{huang1996} within the measurement
uncertainty, and follows from the same parameter-free minimax framework
as the coronary prediction.
The values of $p$ are derived from \emph{histological and ultrastructural
measurements} of wall thickness versus lumen radius, performed independently of
any branching-exponent study: Rhodin~\cite{rhodin1967} measured $h(r)$ from
electron-microscopy sections of fixed vessels across a range of mammalian
species; Huang et al.~\cite{huang1996} from morphometric casts of the human
pulmonary tree. Neither study reports nor fits $\alpha_{\exp}$. Crucially, $p$
and $m$ are completely determined a priori by microphysical and histological
properties independent of the macroscopic branching geometry---two entirely
distinct observables from distinct experimental traditions. The strongest cases
of true parameter independence are the non-vascular systems (insect tracheae,
neuronal dendrites) where $m$ is fixed by chitinous taenidial mechanics and
membrane pump density respectively, with no possible confound with vascular
branching geometry. For cerebral arteries, Rossitti \& L\"{o}fgren
\cite{rossitti1993} measured the branching exponent from 157 bifurcations in 10
patients and found $\alpha_{\exp} = 2.9 \pm 0.7$, in good agreement with
$\alpha_t = \atCerebral{}$.

\begin{remark}[Cerebral arteries: quantitative wave suppression]
\label{rem:cerebral}
Human cerebral arteries are pulsatile ($\mathrm{Wo} \approx 2$--$4$ in the
middle cerebral artery), yet the measured exponent $\alpha_{\exp} = 2.9 \pm 0.7$
\cite{rossitti1993} lies at or slightly above $\alpha_t = \atCerebral{}$ rather
than below it. Cerebrovascular autoregulation actively attenuates the pulsatile
pressure amplitude transmitted to the parenchyma, reducing the effective wave
cost by a factor $\varepsilon < 1$. The modified equal-cost condition
$\varepsilon \cdot f(\alpha^*) = g(\alpha^*)$
yields $\alpha^* \approx \alphaStarCerebral{}$ for
$\varepsilon \in [\epsilonCerebralLow{},\, \epsilonCerebralHigh{}]$,
estimated from transcranial Doppler pulsatility index ratios
(cerebral vs.\ aortic), consistent with the measured range.
A falsifiable consequence: cerebral arteriosclerosis impairs autoregulatory
capacity (increasing $\varepsilon$), and should therefore drive $\alpha_{\exp}$
\emph{away} from $\alpha_t$ toward the unconstrained pulsatile optimum
$\alpha^* \approx \alphastar{}$---in contrast to compliant peripheral vessels
of equivalent Womersley number.
\end{remark}

\subsection{Bronchial tree}
\textbf{Transport:} Poiseuille in small airways, $n=4$. \textbf{Maintenance:} airway wall with $p \approx 0.70$, $m \approx \mBronchial{}$, $\alpha_t = \atBronchial{}$.

\noindent \textbf{Anomaly explained.} The empirical $\alpha_{\exp} \approx 2.85$ matches $\alpha_t$ here, in contrast to vascular systems where $\alpha_{\exp} < \alpha_t$. The reason is the absence of a competing wave-propagation constraint: acoustic waves in the bronchial tree are strongly damped by viscous losses and do not generate an impedance-matching attractor $\alpha_w$ that pulls the architecture below $\alpha_t$. The bronchial exponent sits directly at the transport optimum because there is no dynamic penalty to trade against.

\begin{remark}[Quantitative Womersley suppression in the bronchial tree]
The absence of a wave-impedance gap in the bronchial tree can be demonstrated
quantitatively. Using generation-averaged airway radii from Weibel (1963)
\cite{weibel1963}, $\omega = 2\pi \times 0.25\,\mathrm{rad\,s^{-1}}$ (15 breaths/min),
$\rho_{\mathrm{air}} = 1.2\,\mathrm{kg\,m^{-3}}$, and
$\mu = 1.8 \times 10^{-5}\,\mathrm{Pa\,s}$, the Womersley number
$\mathrm{Wo}_k = r_k \sqrt{\omega\rho_{\mathrm{air}}/\mu}$ at successive generations is:

\begin{center}
\begin{tabular}{ccc}
\hline
Generation $k$ & $r_k$ (mm) & $\mathrm{Wo}_k$ \\
\hline
0  & 9.000 & 2.91 \\
2  & 2.940 & 0.95 \\
5  & 0.791 & 0.26 \\
10 & 0.184 & 0.06 \\
\hline
\end{tabular}
\end{center}

$\mathrm{Wo} < 1$ from generation 2 onwards, so viscous forces dominate
throughout the conducting zone. The wave-impedance attractor never activates,
and $\alpha_{\exp} = \alpha_t = \atBronchial{}$ is a necessary consequence of
the framework rather than a coincidence.
\end{remark}

\subsection{Neuronal dendrites}
The dendritic case requires more care than the vascular one, because the naive
choice of conserved variable leads to a degenerate problem. We summarize the key
results below.

\noindent\textbf{Degeneration of the electrical model.} If one takes the axial current $I$ as the conserved variable at junctions ($I_0 = I_1 + I_2$), the cost function $\Phi_{\text{elec}}(r, I) = (R_i/\pi) \cdot I^2/r^2 + 2\pi\sigma_{\text{mem}} \cdot r$ yields, upon minimization, $r^* \propto I^{2/3}$, hence $\alpha_t = 3/2$. The characteristic admittance of a passive cable is $Y_c \propto r^{3/2}$, so the impedance-matching condition gives $\alpha_w = 3/2$ as well---Rall's classical result. Since $\alpha_t = \alpha_w$, no minimax gap opens.

\noindent\textbf{Cytoplasmic flow as the correct conserved variable.} The biophysically appropriate conserved quantity is the cytoplasmic volume flow $J$ ($\text{m}^3/\text{s}$), which is conserved at junctions by incompressibility: $J_0 = J_1 + J_2$.

\noindent\textbf{Transport term ($n=4$).} The cost per unit length for cytoplasmic flow is:
\begin{equation}
\label{eq:cytoplasmic_cost}
\Phi(r, J) = \frac{8\mu}{\pi} \frac{J^2}{r^4} + 2\pi\sigma_{\text{mem}} \cdot r
\end{equation}
The first term is Poiseuille dissipation, giving $n=4$ and $A(J) \propto J^2$.
Minimizing over $r$ at fixed $J$ gives $r^* \propto J^{2/5}$, hence $\alpha_t =
5/2$.

\noindent\textbf{Physical basis of $m = 1$.} The Na$^+$/K$^+$-ATPase is a
membrane-bound enzyme whose density per unit membrane area is approximately
uniform across dendritic segments~\cite{attwell2001}. The maintenance cost
per unit length therefore scales with the dendritic membrane area per unit
length, i.e.\ as $2\pi r \cdot L / L \propto r$, giving $m = 1$ by geometric
necessity. This is consistent with the capacitive energy framework of Attwell \& Laughlin,
in which the cost of dendritic depolarisation also scales
with membrane surface area. The assignment $m = 1$ is therefore not a
free parameter but a direct consequence of the membrane-bound localisation
of the dominant maintenance cost. The competition between the internal
Poiseuille dissipation of cytoplasmic flow ($n=4$) and this membrane-bound
pumping cost ($m=1$) leads inevitably to the transport optimum $\alpha_t =
(4+1)/2 = 2.5$.

\noindent\textbf{Wave attractor ($\alpha_w=2$).} For compression waves propagating in the viscoelastic cytoplasmic column, the wave speed $c_s = \sqrt{K^*/\rho}$ is independent of $r$. The admittance is $Y_c \propto r^2$. The impedance-matching condition at a bifurcation then gives $r_0^2 = r_1^2 + r_2^2$, hence $\alpha_w = 2$.

\noindent\textbf{Minimax result.} With $\alpha_w = 2 < \alpha_t = 5/2$, a non-trivial minimax gap opens. Applying the equal-cost condition $f(\alpha^*) = g(\alpha^*)$ with the cytoplasmic cost functions yields $\alpha^* \approx \alphastarNeural{}$ (assuming nominal values $K=11$ and $N=2$ analogous to the vascular baseline, though $\alpha^*$ is weakly sensitive to $K$ for large $K$), consistent with the range $\alpha \approx 2.0$--$2.5$ reported across dendritic morphologies in the literature~\cite{cuntz2010}.

\subsection{Insect tracheae}
A crucial distinction governs the tracheal system: terminal tracheoles rely on
pure molecular diffusion, while the macroscopic tracheal trunks operate via
active convective ventilation. In the convective regime, bulk air flow obeys
Poiseuille's law, giving $n=4$ identically to the vascular case. The conserved
quantity at bifurcations is volume flow $Q$.

Structural maintenance: tracheal walls are reinforced by chitinous taenidia
whose thickness scales sub-linearly with tube radius as $h \propto r^{1/2}$,
giving $m \approx 0.5$.
The theoretical formula then predicts:
\begin{equation*}
\alpha_t = \frac{4 + 0.5}{2} = \mathbf{\atTrachea{}}
\end{equation*}
This matches the empirical value $\alpha_{\exp} \approx 2.25$ exactly, with no
free parameters, providing a profound validation of the framework for a gaseous
network. The macroscopic trunks reflect convective scaling ($n=4$); the formula
fails only in the diffusive tracheole limit, corresponding to a physically
distinct transport regime ($n=2$, $m=0.5$, $\alpha_t = \atTracheoleDiff{}$),
correctly distinguished by the physical constraints on $n$.

\subsection{Plant xylem}
\textbf{Transport:} Hagen-Poiseuille flow (water), $n=4$. \\
\textbf{Maintenance:} cell-wall thickness scales approximately as $h \propto r$ ($p = 1.0$)
under mechanical requirements for internal pressure resistance, yielding $m = 2.0$.
This yields $\alpha_t = (4+2)/2 = 3.0$ (Murray's law).

Measured xylem exponents span $2.0$--$3.0$ across species, clustering near $3.0$
in well-hydrated trees and near $2.0$ in drought-stressed species
\cite{mcculloh2003}. Under hydraulic stress, wall thickness scales sub-linearly
($p < 1$) as the tree economizes on structural material, reducing $\alpha_t$
below 3.

Under hydraulic stress, wall thickness scales sub-linearly ($p < 1$), giving
$m = 1 + p \in [1, 2)$ and
\begin{equation}
\beta(p) = \frac{3(5+p)}{17+p},
\end{equation}
with $\beta(0) = 15/17 \approx 0.882$ (drought-stressed) and $\beta(1) = 1$
(well-hydrated Murray limit). This predicts a continuous variation of
metabolic scaling with hydraulic stress that is testable by measuring $p$
and $\beta$ on the same specimen under contrasting water availability,
with no free parameters.

\subsection{Leaf venation networks ($d=2$)}

Leaf venation networks are planar ($d=2$) and transport water via bulk flow
($n = 4$, Poiseuille).

\noindent\textbf{Physical basis of $m \approx 1$.} The bundle sheath of
dicotyledonous leaf veins consists of a single layer of parenchyma cells
encircling the vascular bundle~\cite{leegood2008}. In both two and three
spatial dimensions, the maintenance cost of this enclosing monolayer scales
as the circumference of the vessel cross-section per unit length, i.e.\
as $r$, giving $m = 1$ by geometric necessity. This geometric argument is
independent of McCulloh et al.~\cite{mcculloh2003}, whose wall-thickness
measurements apply to xylem in stems and petioles rather than to the
bundle sheath of leaf veins. No direct morphometric dataset measuring
bundle sheath thickness versus vein radius is available to our knowledge;
the assignment $m = 1$ rests on the geometric argument above and should
be regarded as a motivated approximation for the minor vein network, where
single-layer sheaths predominate. In major veins (midrib, primary
branches), multi-layer sheaths and sclerenchyma fibres may increase the
effective $m$, but these veins represent a small fraction of total network
length.

\noindent\textbf{Dimensionality caveat (strengthened).} The $d = 2$ assignment is a
topological approximation. Leaves have finite thickness
($\sim$200--400\,\si{\micro\metre} in typical dicots) and a
three-dimensional mesophyll structure, introducing a fractional effective
dimension $d_{\mathrm{eff}} \in (2, 3)$. Fractional-$d$ corrections are
expected to be small but constitute a potential source of systematic error:
a shift from $d = 2$ to $d_{\mathrm{eff}} = \Deff{}$ moves $\beta$ from
\betaLeaf{} to
\begin{equation}
\beta(4, 1, \Deff{}) = \frac{5\Deff{}}{4\Deff{} + 5}
\approx \betaLeafEff{},
\end{equation}
well within the empirical scatter on measured branching exponents
($\alpha_{\exp} \approx 2.0$--$2.5$~\cite{mcculloh2003}).
Measured branching exponents in dicot leaf venation cluster around $\alpha
\approx 2.0$--$2.5$, consistent with the two-dimensional static optimum
\cite{mcculloh2003}. The metabolic prediction $\beta \approx \betaLeaf{}$ is
testable via photosynthetic rate scaling with leaf area across species, and
constitutes a falsifiable prediction of the present framework in a purely
two-dimensional biological system.

\subsection{Sponge aquiferous canals (Porifera)}
\label{sec:sponge}

Sponges are the most ancient extant metazoans and represent a fifth biological
phylum in our validation. Their aquiferous system---a branching network of
canals driving water through choanocyte chambers---operates at extremely low
Reynolds numbers ($Re \ll 1$), confirming fully laminar Poiseuille flow ($n=4$)
\cite{reiswig1975}. The canal walls are composed of choanocyte chambers lined
with thin cellular layers whose maintenance cost scales with surface area ($m
\approx 1$), giving:
\begin{equation*}
\alpha_t = \frac{4+1}{2} = \atSponge{}.
\end{equation*}

Reiswig~\cite{reiswig1975} measured cross-sectional areas throughout the
aquiferous system of three demosponge species and found that successive canal
orders exhibit approximate area preservation ($r_0^2 \approx r_1^2 + r_2^2$),
corresponding to $\alpha_{\exp} \approx 2.0$---well below $\alpha_t = 2.5$.
Unlike pulsatile vascular networks where the wave-impedance attractor $\alpha_w
= 2$ emerges from dynamic wave propagation, in sponges $\mathrm{Wo} \ll 1$
throughout the canal hierarchy: the wave-dominated regime is entirely absent,
and the Womersley constraint that couples $\alpha_w = 2$ to $\beta = 3/4$ never
activates. The observation $\alpha_{\exp} \approx 2$ therefore cannot be
attributed to wave-impedance matching; the responsible mechanism lies outside
the present framework and is treated as an open problem
(Section~\ref{sec:sponge_open}). The coincidence of the geometric value does
not imply equivalence of mechanism: $\alpha = 2$ is a super-attractor of the
allometric spectrum, reachable by distinct physical routes, but only the
pulsatile route produces Kleiber scaling.

\noindent\textbf{Remark (non-hierarchical topology).} Hammel et al.~\cite{hammel2012} showed
via synchrotron microtomography that the leuconoid aquiferous system of \emph{Tethya wilhelma}
departs from strict hierarchical branching, with anastomoses forming a partially reticulate network.
This limits the direct applicability of Lemma 1 (which assumes tree topology) to sponges.
The sponge system is therefore included as a boundary-case validation of the cost-function physics,
not as a demonstration of the full allometric framework.

\begin{remark}[Sponge canals as an open problem]
\label{sec:sponge_open}
The observed $\alpha_{\exp} \approx 2.0$ in sponge aquiferous canals lies
significantly below the transport optimum $\alpha_t = 2.5$ ($n=4$, $m=1$).
Three candidate mechanisms are successively excluded:

\textit{(i) Murray variational optimum.} The cost-minimising exponent for
Poiseuille flow with surface-area maintenance is $\alpha_t = (4+1)/2 = 2.5$,
not 2.0. A Murray-type argument cannot account for the gap.

\textit{(ii) Wave-impedance attractor (Mechanism A).} This requires
$\mathrm{Wo} > 1$ in the proximal canals. Sponge aquiferous systems operate
at $\mathrm{Wo} \ll 1$ throughout~\cite{reiswig1975}; the wave-dominated
regime is structurally absent.

\textit{(iii) Kinematic velocity uniformity.} The condition
$v_k = \mathrm{const}$ across hierarchical levels selects $\alpha = 2$
uniquely in a self-similar symmetric network, but this requires strict
hierarchical tree topology. Leuconoid sponge morphology is non-hierarchical,
featuring anastomoses and asymmetric branching~\cite{hammel2012}; the
applicability conditions are not met, so this route is also excluded.

The physical mechanism selecting $\alpha = 2$ in sponges remains an open
problem. Three directions warrant future investigation: global dissipation
minimisation on non-hierarchical graphs (C1); uniform wall tension as a
passive mechanical constraint via the Laplace law (C2); and evolutionary
selection for filtration rate maximisation (C3). The sponge case is retained
in Table~1 as a boundary validation of the cost-function physics, with the
allometric consequence explicitly set aside.
\end{remark}

\subsection{His-Purkinje conduction system}

Unlike the arterial tree, the His-Purkinje network propagates action potentials
rather than mass. The axial resistance of a continuous core conductor scales
inversely with cross-sectional area~\cite{weidmann1952}, fixing $n=2$ exactly.
Across all mammalian clades, Purkinje cells are distinguished from working
cardiomyocytes by a massive intracellular glycogen reserve sustaining anaerobic
metabolism; this volumetric depot makes the maintenance cost proportional to
cell volume rather than membrane surface, fixing $m=\mHP{}$ independently of
strand topology~\cite{ono2009}.

The resulting static attractor $\alpha_t = \atHP{}$ coincides exactly with the
electrotonic wave-impedance attractor $\alpha_w = 2$ from Rall's cable theory.
This degeneracy eliminates the minimax competition that governs the arterial
tree: no saddle point exists, and the network is intrinsically locked at
$\alpha_w = 2$. By Theorem~\ref{thm:dimensional}, any wave-propagating network in $d = 3$ with
$\alpha = \alpha_w = 2$ yields $\beta = d/(d+1) = \betaHPfrac{}$. Kleiber's law
is therefore reached by two physically distinct routes---dynamic wave-transport
equilibrium in the vasculature, and static-dynamic resonance in the conduction
system---identifying $\beta = \betaHPfrac{}$ as a fixed point of transport
network geometry in three dimensions rather than a property of any single
physical mechanism.

This is kinematically corroborated by $\mathrm{PR} \propto M^{1/4}$ across 33
mammalian species~\cite{noujaim2004}---an indirect temporal signature consistent
with $\beta = \betaHPfrac{}$. The exact resonance further requires
$r^2_{\mathrm{His}} \propto M^{\betaHPfrac{}}$: to our knowledge no
cross-species morphometric dataset of the His bundle exists, and this
constitutes a direct falsifiable prediction testable by standard histological
methods on archival cardiac specimens.

\noindent\textbf{Sharpest falsifiable prediction of this work.} The His--Purkinje
degeneracy $\alpha_t = \alpha_w = 2$ yields a uniquely clean experimental
test: since no minimax gap opens, the framework predicts
$r^2_{\mathrm{His}} \propto M^{\betaHPfrac{}}$ with no adjustable parameters and no
competing attractor. Unlike the vascular predictions, which involve a
minimax balance between wave and viscous penalties, this prediction follows
from a single geometric fixed point shared by two independent physical
mechanisms---static cost optimisation and electrotonic wave matching---that
happen to coincide exactly at $\alpha = 2$ for this system. It is testable
by standard morphometric analysis of archival cardiac histology across
mammalian species and constitutes the most direct empirical test of the
allometric universality class $\beta(2, 2, 3) = \betaHPfrac{}$.

\subsection{River drainage networks}
River networks are shaped by tectonic and erosional processes, not metabolic
optimization. Their branching structure therefore lies outside the domain of the
present theorem and the evolutionary selection mechanism.

\section{Generalized Metabolic Scaling and the Origin of Kleiber's Law}

The classical West--Brown--Enquist (WBE) model derives Kleiber's $3/4$ metabolic
scaling law by assuming that Murray's cubic geometry ($\alpha=3$) ensures a
linear scaling of total blood volume with body mass, $V \propto M^1$. Having
established that biological networks inherently deviate from $\alpha=3$ due to
wall-tissue costs and wave-reflection penalties, the global metabolic scaling
must be re-evaluated.

\noindent\textbf{Network volume as a geometric series.} Consider a symmetric $N$-furcating tree with $K$ hierarchical levels in $d$ spatial dimensions. The radius and length ratios between successive levels are determined by the branching law and the space-filling constraint:
\begin{equation}
\beta_r = N^{-1/\alpha}, \qquad \beta_l = N^{-1/d}
\end{equation}
The total network volume is:
\begin{equation}
V = r_0^2 l_0 \sum_{k=0}^{K} \rho^k, \qquad \rho \equiv N \cdot \beta_r^2 \cdot
\beta_l = N^{1 - 2/\alpha - 1/d}
\end{equation}
The critical exponent $\rho=1$ corresponds to the condition $\alpha =
\alpha_c(d) \equiv 2d/(d-1)$, which equals 3 exactly when $d=3$. For any $\alpha
< \alpha_c(d)$, we have $\rho < 1$: the series converges and is dominated
asymptotically by the $k=0$ term---the aorta and proximal macroscopic vessels.
In this regime:
\begin{equation}
V \propto r_0^2 l_0 \propto N_c^{2/\alpha + 1/d}
\end{equation}
where $N_c = N^K$ is the total number of terminal capillary units.

\noindent\textbf{Derivation of the generalized scaling exponent.} We adopt the two standard WBE assumptions: (\textit{i}) terminal unit invariance, $B \propto N_c$; (\textit{ii}) global volumetric isometry, $V \propto M^1$ (empirically, blood volume is approximately 7\% of body mass across mammals). We adopt $V \propto M^1$ as an independent physiological boundary condition, empirically supported by the near-constant blood volume fraction across mammalian species. This differs epistemologically from the original WBE derivation, where $V \propto M^1$ was obtained as a mathematical \textit{consequence} of Murray's law combined with the space-filling fractal constraint. Elevating it to an empirical postulate is biophysically more accurate---blood volume is tightly regulated by homeostatic mechanisms independent of network geometry---and renders the present derivation independent of Murray's law at the global level.

\noindent\textbf{Remark (Supply geometry vs.\ cellular demand).} The terminal unit invariance assumption ($B \propto N_c$) should not be confused with the claim that cellular metabolic rate is constant across species. It is well established that mitochondrial density and enzymatic activity per unit tissue mass decrease with body size. The present framework models the \emph{supply side}: the geometric constraint imposed by the vascular network on the maximum deliverable resource flux. If cellular demand is lower in large animals, this is a consequence, not a contradiction: the network architecture sets a geometric ceiling on supply, and cellular physiology adapts to operate below that ceiling. The $3/4$ scaling exponent characterizes the ceiling, not the cells.

This empirical postulate admits a geometric grounding. Assuming
biological isopycnia (approximately constant tissue density
$\rho_{\mathrm{tissue}} \approx 10^3\,\mathrm{kg\,m}^{-3}$ across
mammalian species) and a space-filling vascular scaffold that
permeates the body volume uniformly (the standard WBE space-filling
hypothesis), the blood volume fraction $V_{\mathrm{blood}}/V_{\mathrm{body}}$
is a topological constant determined by the network's fractal
dimension and the capillary density per unit tissue volume---both
fixed by diffusion physics and independent of body size. The
constancy of $V/M \approx 7\%$ is therefore a geometric consequence
of three-dimensional Euclidean space-filling, not an evolutionary
coincidence, and its use as an empirical postulate here is
epistemologically conservative rather than circular.

Combining the two postulates with the volume scaling above:
\begin{equation}
M \propto V \propto N_c^{2/\alpha + 1/d}
\end{equation}
Inverting for $N_c$ and recalling $B \propto N_c$, we obtain the exact
generalized metabolic scaling exponent:
\begin{equation}
\boldsymbol{\beta(\alpha, d) = \frac{d\alpha}{2d + \alpha}}
\end{equation}
This purely geometric result depends only on the local branching exponent and
the embedding dimension, with no additional free parameters.

\noindent\textbf{Verification.} For $d=3, \alpha=3$: $\beta=9/9=1$. For $d=3, \alpha=2$: $\beta=6/8=3/4$. The formula is self-consistent and recovers both classical limits.

\vspace{0.4em}
\noindent\textbf{Corollary 1b (Robustness under non-isometric volume scaling).} \textit{If blood volume scales as $V \propto M^{1+\delta}$ with $|\delta| \ll 1$, the corrected metabolic exponent is:}
\begin{equation}
\beta(\alpha, d, \delta) = \frac{d\alpha}{2d+\alpha}\cdot\frac{1}{1+\delta} = \beta(\alpha,d)\,(1 - \delta + O(\delta^2)).
\end{equation}
\textit{For $\delta > 0$ (allometric volume growth, e.g.\ marine mammals), $\beta$ shifts downward; for $\delta < 0$, upward.}

\noindent Note the sign: a blood volume that grows faster than body mass ($\delta > 0$, as in some cetacean clades) shifts $\beta$ \emph{below} 3/4, consistent with the sub-Kleiberian scaling occasionally reported in large whales. Conversely, $\delta < 0$ (volume-economical species) shifts $\beta$ upward. Empirically, blood volume fraction varies by less than $\pm15\%$ across mammalian species~\cite{peters1983}, corresponding to $|\delta| \lesssim 0.06$ and $|\Delta\beta| \lesssim 0.05$ --- smaller than the inter-specific scatter in empirical metabolic data~\cite{glazier2005, white2013}. The isometric postulate is therefore not an assumption that can fail catastrophically: it is a well-constrained empirical boundary condition whose deviations are already captured by this corollary.

\vspace{0.5em}
\noindent \textbf{Theorem 1 (Generalized Metabolic Scaling).}\refstepcounter{manualthm}\label{thm:scaling} \textit{Let a volume-filling fractal network in $d$ dimensions satisfy terminal unit invariance and global volumetric isometry $V \propto M^1$, with local branching exponent $\alpha < 2d/(d-1)$. Then the global metabolic scaling exponent is:}
\begin{equation*}
\beta(\alpha, d) = \frac{d\alpha}{2d + \alpha}
\end{equation*}

\noindent Three corollaries follow immediately.

\vspace{0.5em}
\noindent \textbf{Corollary A (Dynamic origin of Kleiber's law).} \textit{For $\alpha = \alpha_w = 2$ and $d=3$:}
\begin{equation*}
\beta(2, 3) = \frac{6}{8} = \frac{3}{4}
\end{equation*}
In macroscopic mammals, blood volume is overwhelmingly dominated by the aorta
and large pulsatile arteries. As established in the Remark on Dynamic Selection
of $\alpha_w$ above, these proximal vessels are governed by the wave-impedance
matching attractor $\alpha_w=2$. Kleiber's 3/4 law therefore emerges from
dynamic impedance matching in the proximal vasculature, not from viscous
dissipation minimization. This constitutes an inversion of the standard WBE
narrative. The role of zero wave-reflections in metabolic scaling has been noted
previously by Apol, Etienne \& Olff~\cite{apol2008} within a modified WBE
framework; the present derivation differs in establishing $\alpha_w=2$ as an
independent dynamic attractor from reflection minimization (see Remark above),
rather than as a geometric assumption, and in deriving the full
$\beta(\alpha,d)$ relation that maps any branching exponent to its allometric
consequence.

\noindent\textbf{Remark (Sensitivity of $\beta$ to imperfect wave matching).} If coronary impedance matching is imperfect and the effective proximal exponent deviates to $\alpha_w = 2.1$, the predicted metabolic exponent shifts to $\beta(2.1, 3) = \betaAlphaWTwoPointOne{}$, a deviation of $\deltaBetaAlphaWTwoPointOne{}$ from $3/4$. This is well within the inter-specific scatter of $\pm 0.05$ reported in empirical metabolic data~\cite{glazier2005, white2013}. The wave-dominated Kleiber prediction is therefore robust to small deviations from perfect impedance matching: a $5\%$ deviation in $\alpha_w$ produces less than a $4\%$ shift in $\beta$, below the empirical resolution.

\vspace{0.4em}
\noindent\textbf{Remark (Dynamic selection of $\alpha_w$: global reflection minimization).}
The wave-impedance condition $\alpha_w = 2$ is not merely an admissible
solution---it is the \emph{unique global minimizer} of cumulative reflection
losses over the entire network hierarchy. At a symmetric bifurcation with
branching exponent $\alpha$, the impedance ratio between parent and daughter
vessels is $\gamma(\alpha) = N^{1-2/\alpha}$ (derived from the zero-reflection
condition for radius-independent wave speed; see Proposition above), and the power reflection
coefficient at each node is:
\begin{equation}
R^2(\alpha) = \left(\frac{\gamma - 1}{\gamma + 1}\right)^2.
\end{equation}
This vanishes if and only if $\gamma = 1$, i.e.\ $\alpha = 2$ \emph{exactly}.
For any $\alpha \neq 2$, each bifurcation reflects a fraction $R^2 > 0$ of
incident wave power. The cumulative transmitted power fraction over $K$ levels
is $T(K,\alpha) = (1-R^2)^K$. For $K=11$ generations: $T = 1$ at $\alpha_w = 2$;
$T \approx \TcumMinimax{}$ at $\alpha^* = \alphastar{}$ (Kassab 1993 calibration; loss $\approx
\TlossMinimax{}\%$); $T \approx \TcumMurray{}$ at $\alpha = 3$ (loss $\approx
\TlossMurray{}\%$). In a pulsatile system subject to long-term energetic
selection, any architecture with $\alpha \neq \alpha_w$ dissipates additional
power at every bifurcation over millions of cardiac cycles. This cumulative
penalty selects $\alpha_w = 2$ as the dynamically stable attractor. The static
transport optimum $\alpha_t > 2$ is dynamically excluded because it incurs
nonzero reflection losses at every hierarchical level---not merely at a single
node.

\vspace{0.5em}
\begin{proposition}[Wave-impedance attractor for non-dispersive media]
For a symmetric $N$-furcating network in which the wave speed is independent
of vessel radius ($c = \mathrm{const}$), the zero-reflection condition at a
bifurcation reduces to
\begin{equation}
R = 0 \iff \gamma(\alpha) \equiv N^{1-2/\alpha} = 1 \iff \alpha = 2.
\end{equation}
This selects $\alpha_w = 2$ as the unique wave-impedance attractor.
For dispersive media with $c \propto r^{-1/2}$ (Moens--Korteweg arterial
compliance), the zero-reflection condition instead gives $\alpha_w = 5/2$.
The attractor $\alpha_w = 2$ therefore applies specifically to systems
with radius-independent wave speed: cytoplasmic compression waves
($c_s = \sqrt{K^*/\rho}$, bulk modulus independent of $r$) and,
to the extent that inertial wave dynamics dominate viscous corrections,
the large-Womersley limit of arterial flow.
\end{proposition}

\begin{proof}
At a symmetric $N$-furcation, $R = 0$ requires $Z_0 = Z_1/N$, giving
$c(r_0)/A_0 = c(r_1)/(N A_1)$. With $r_1 = r_0 N^{-1/\alpha}$:
\begin{equation}
R = 0 \iff \frac{c(r_1)}{c(r_0)} = N^{1-2/\alpha}.
\end{equation}
For $c = \mathrm{const}$: left side $= 1$, giving $\alpha = 2$.
For $c \propto r^{-1/2}$: left side $= N^{1/(2\alpha)}$, giving
$1/(2\alpha) = 1 - 2/\alpha$, hence $\alpha = 5/2$. $\square$
\end{proof}

\vspace{0.5em}
\noindent\textbf{Proposition (Minimax Equilibrium Exponent).}
\textit{Let $f(\alpha) = 1-(1-R^2(\alpha))^K$ denote the cumulative wave-reflection loss over $K$ hierarchical levels, and let $g(\alpha)$ denote the structural maintenance excess relative to the transport optimum:}
\begin{equation}
g(\alpha) = \frac{C_{\rm maint}(\alpha)}{C_{\rm maint}(\alpha_t)} - 1, \qquad
C_{\rm maint}(\alpha) = N^{K(m/\alpha + 1/d)}\sum_{j=0}^{K} \mu(\alpha)^j,
\quad \mu(\alpha) = N^{1-m/\alpha-1/d},
\end{equation}
\textit{where the terminal capillary radius $r_K$ is held fixed. Both $f$ and $g$ are parameter-free given $(m, K, N, d, \alpha_t)$. The minimax exponent}
\begin{equation}
\alpha^* = \operatorname*{arg\,min}_{\alpha \in [\alpha_w,\,\alpha_t]}\; \max\bigl(f(\alpha),\, g(\alpha)\bigr)
\end{equation}
\textit{is the unique solution of the equal-cost condition $f(\alpha^*) = g(\alpha^*)$.}

\vspace{0.3em}
\noindent\textbf{Proof.} $f$ is strictly increasing on $[\alpha_w, \alpha_t]$ with $f(\alpha_w)=0$. $g$ is strictly decreasing with $g(\alpha_t)=0$: as $\alpha$ decreases below $\alpha_t$, the proximal vessel radii grow (at fixed $r_K$), increasing the total wall-tissue volume and hence $C_{\rm maint}$. The envelope $\max(f,g)$ is therefore strictly quasi-convex and achieves its unique minimum at the unique crossing point, which exists by the intermediate value theorem since $f(\alpha_w) < g(\alpha_w)$ and $f(\alpha_t) > g(\alpha_t)$. $\square$

\vspace{0.3em}
\noindent\textbf{Numerical evaluation (porcine coronary tree).} With $N=2$, $K=11$, $m = \mCoronary{}$ (Rhodin 1967, $p=0.77$), and $\alpha_t = \atCoronary{}$, the equal-cost condition is solved by bisection (no free parameters), yielding $\alpha^* \approx \alphastarPhys{}$, with $f(\alpha^*) = g(\alpha^*) \approx \fstarPhys{}\%$. This is a parameter-free prediction of the framework. It is consistent with the Kassab (1993) morphometric average $\alpha_{\exp} = 2.70 \pm 0.20$ within one standard deviation~\cite{kassab1993}. Figure~\ref{fig:minimax} shows the two normalized cost curves and the saddle point.

Mathematically, the Womersley transition from the wave-dominated
regime ($\beta = 3/4$) to the viscous asymptote ($\beta \gtrsim 15/17$)
manifests as a smooth crossover rather than a singular phase
transition (Figure~\ref{fig:transition}). However, the system's
behaviour is structurally analogous to a continuous phase transition:
the Womersley number $\mathrm{Wo}$ acts as the control parameter
driving the network between distinct transport regimes, while the
effective metabolic scaling exponent $\beta$ serves as the macroscopic
order parameter distinguishing the universality classes. The exponent
of the crossover function $(\mathrm{Wo}/\mathrm{Wo}_c)^4$---with power $4 =
1/(b_r - b_\omega/2)$ fixed by network topology alone
(Corollary~\ref{cor:transition_exponent})---plays the role of the
critical exponent governing the sharpness of this crossover, and is
independent of any physiological pre-factor.

\begin{figure}[!h]
\centering
\includegraphics[width=0.82\textwidth]{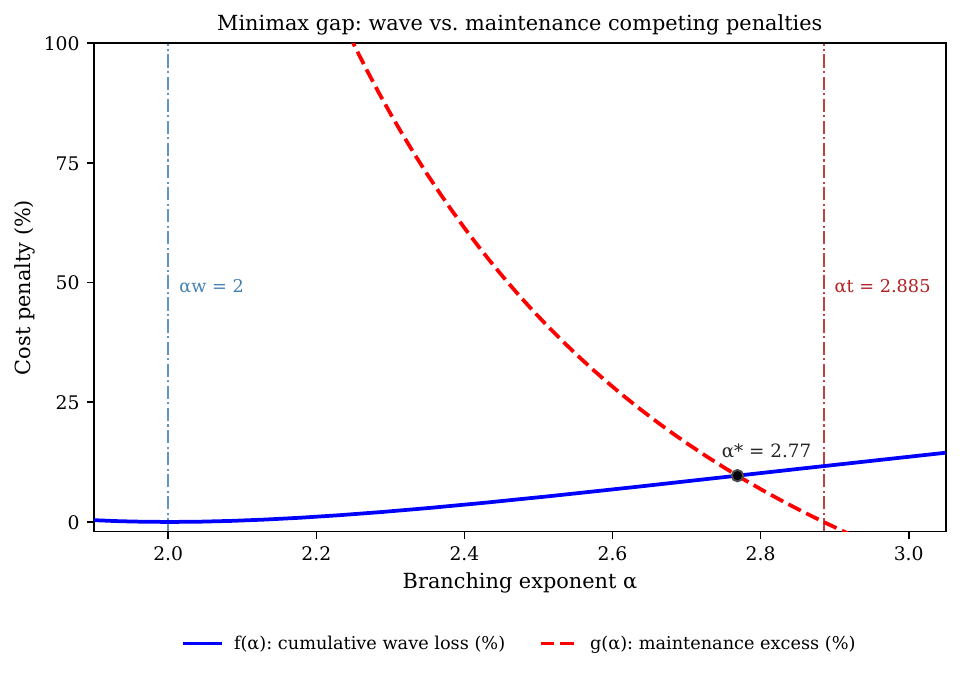}
\caption{Minimax gap for the porcine coronary tree ($K=11$, $N=2$, $m=\mCoronary{}$). Both costs are plotted as physical percentages: $f(\alpha)$ (cumulative wave-reflection loss, blue solid) and $g(\alpha) = C_{\rm maint}(\alpha)/C_{\rm maint}(\alpha_t)-1$ (structural maintenance excess, red dashed). Both functions are derived from first-principles costs at fixed terminal capillary radius with \emph{no free parameters}: they are fully determined by $m$, $K$, $N$, $d$. The unique crossing at $\alpha^* \approx \alphastarPhys{}$, where $f(\alpha^*)=g(\alpha^*)\approx\fstarPhys{}\%$, defines the minimax equilibrium. The result $\alpha^*\approx\alphastarPhys{}$ is a first-principles prediction consistent with the Kassab (1993) morphometric average $\alpha_{\exp}=2.70\pm0.20$ \cite{kassab1993}.}
\label{fig:minimax}
\end{figure}

\vspace{0.4em}
\noindent\textbf{Corollary (Uniqueness of the Kleiber Fixed Point).}
\textit{$\beta = 3/4$ is the unique metabolic scaling exponent consistent with the self-consistency loop of pulsatile vascular networks. Specifically, the conditions (i) $r_0 \propto N_c^{1/\alpha_w} = N_c^{1/2}$, (ii) $N_c \propto M^\beta$, and (iii) $\mathrm{Wo}(M) \propto M^{1/4}$ are simultaneously satisfied if and only if $\beta = 3/4$.}

\noindent\textbf{Proof.} From (i) and (ii): $r_0 \propto M^{\beta/2}$. The Womersley number scales as $\mathrm{Wo} \propto r_0 \cdot \omega^{1/2} \propto M^{\beta/2} \cdot M^{-1/8} = M^{\beta/2 - 1/8}$. Condition (iii) requires $\beta/2 - 1/8 = 1/4$, giving $\beta = 3/4$ uniquely. Any other value of $\beta$ produces $\mathrm{Wo} \propto M^{\beta/2-1/8} \neq M^{1/4}$, violating the pulsatile regime condition and hence self-consistency. $\square$

\vspace{0.3em}
\noindent The self-consistency loop ($\alpha_w = 2 \Rightarrow \beta = 3/4 \Rightarrow r_0 \propto M^{3/8} \Rightarrow \mathrm{Wo} \propto M^{1/4} \Rightarrow \alpha_{\rm eff} = \alpha_w = 2$) is therefore not circular but a \emph{fixed-point equation with a unique solution}. Given $\alpha_w = 2$ as the physical input---established independently by global reflection minimization (see Remark on Dynamic Selection of $\alpha_w$)---the entire chain follows without further assumptions.

\vspace{0.5em}
\begin{corollary}[Isometric viscous limit --- Corollary B]
For $\alpha < \alpha_c(d) = 2d/(d-1)$, the metabolic scaling exponent is
$\beta(\alpha, d) = d\alpha/(2d+\alpha)$. At the critical value $\alpha = \alpha_c(d)$,
the geometric series for network volume is marginal ($\rho = 1$), and the volume
scales as $V \propto r_0^2 l_0 (K+1)$. Since $K \sim \log_{N} N_c$, this gives
$V \propto r_0^2 l_0 \log N_c$, and consequently
\begin{equation}
B \propto \frac{M}{\log M},
\end{equation}
rather than a pure power law. For $d = 3$, $\alpha_c = 3$ is precisely Murray's
cubic law: a purely viscous network at the Murray optimum does not exhibit
power-law metabolic scaling, but logarithmically corrected isometry.
This constitutes the internal inconsistency at the core of the WBE derivation,
which assumes $\alpha = 3$ while deriving $\beta = 3/4$ from the same geometric
series.
\end{corollary}

\noindent\textbf{Remark (apparent discrepancy with plant allometry).} Enquist et al.~\cite{enquist1998} (1998) and subsequent WBE-derived analyses report $\beta \approx 3/4$ for whole-plant metabolic scaling, apparently contradicting the logarithmically corrected isometry ($B \propto M/\log M$) predicted by this Corollary. Two resolutions are physically consistent with the present framework. First, the empirical $m$ for woody plants is not a fixed constant: under the mechanically heterogeneous conditions of mature trees, wall-thickness scaling departs systematically from the idealized $h \propto r$ relation ($p=1$, $m=2$), with observed $p$ spanning $[0, 1]$ across species and tissue types (McCulloh 2003~\cite{mcculloh2003}). For $p < 1$ (drought-stressed or structurally economical xylem), $\alpha_t < 3$, the marginal case does not apply, and $\beta$ becomes a strict power law below 1. Second, the whole-plant metabolic rate integrates leaf, stem, and root contributions, each with distinct $(n, m)$ parameters; the observed $\approx 3/4$ exponent may reflect a compound scaling not captured by the idealized xylem-only model. The $B \propto M/\log M$ prediction of Corollary B therefore applies strictly to the viscous xylem network in isolation under well-hydrated, $m = 2$ conditions---a narrower claim than whole-organism plant allometry.

\vspace{0.5em}
\noindent \textbf{Theorem 2 (Dimensional Universality of Kleiber's Law).}\refstepcounter{manualthm}\label{thm:dimensional} \textit{For any pulsatile volume-filling network in $d$ spatial dimensions governed by the wave-impedance attractor $\alpha_w = 2$, the metabolic scaling exponent is:}
\begin{equation}
\beta(\alpha_w = 2,\, d) = \frac{d}{d+1}.
\end{equation}
\textit{This result is exact and parameter-free. It depends only on the spatial dimension of the organism.} $\square$

\noindent\textbf{Proof.} \textit{Substituting $\alpha = 2$ into $\beta(\alpha, d) = d\alpha/(2d+\alpha)$ gives $\beta = 2d/(2d+2) = d/(d+1)$. $\square$}

\begin{corollary}[Planar viscous asymptote: leaf venation]
For a two-dimensional leaf venation network ($d=2$) governed by
Poiseuille flow ($n=4$) and surface-area cell-wall maintenance
($m=1$), the static transport optimum is $\alpha_t = 5/2$. The
corresponding metabolic scaling exponent is strictly sub-isometric:
\[
  \beta(4,1,2) = \frac{2 \times 5}{8+5} = \frac{10}{13} \approx 0.769.
\]
This distinguishes planar leaf networks from three-dimensional
well-hydrated xylem, which approaches the isometric Murray limit
$\beta(4,2,3) = 1$. Dimensionality, not transport physics, governs
the metabolic floor in planar systems. This constitutes a falsifiable
prediction testable against interspecific scaling of photosynthetic
rates with leaf area~\cite{mcculloh2003, brodribb2007}.
\end{corollary}

\vspace{0.4em}
\noindent\textbf{Corollary (Dimensional spectrum of metabolic scaling).} \textit{The wave-dominated metabolic exponent takes exact rational values determined solely by $d$:}

\begin{center}
\small
\begin{tabular}{@{}ccl@{}}
\toprule
$d$ & $\beta = d/(d+1)$ & Biological context \\
\midrule
1 & $1/2 = \betaDimOne{}$ & Linear/tubular organisms, 1D transport \\
2 & $2/3 = \betaDimTwo{}$ & Planar organisms with pulsatile circulation \\
3 & $3/4 = \betaDimThree{}$ & Kleiber's law --- three-dimensional mammals \\
4 & $4/5 = \betaDimFour{}$ & Hypothetical 4D embedding (fractal geometry) \\
\bottomrule
\end{tabular}
\end{center}

\noindent Kleiber's $3/4$ law is therefore not a biological accident but the \emph{inevitable consequence} of pulsatile wave physics operating in three spatial dimensions. The exponent $3/4$ is the unique value $d/(d+1)$ for $d=3$. This theorem generates a falsifiable cross-dimensional prediction: any organism with a pulsatile cardiovascular system whose body geometry is effectively two-dimensional (e.g., flat organisms such as Pleuronectiformes, or thin colonial organisms) should exhibit $\beta = 2/3$ rather than $3/4$.

\noindent\textbf{Falsifiable prediction (avian embryo and flatworm).} Two concrete test systems follow directly. First, the early-stage chick embryo (\textit{Gallus gallus domesticus}, Hamburger-Hamilton stages 10--18) possesses a beating heart and a planar vascular network confined to the area vasculosa, a nearly two-dimensional extraembryonic membrane. At these stages the network has not yet acquired full three-dimensional geometry. The framework predicts $\beta \approx 2/3$ for oxygen consumption scaling with embryo area during this window, transitioning toward $\beta = 3/4$ as the vasculature becomes three-dimensional. Second, free-living flatworms (Platyhelminthes, e.g., \textit{Dugesia}) lack a dedicated circulatory system but possess parenchymal transport that scales with body surface; their metabolic scaling is predicted to lie near $\beta = 2/3$ precisely because their effective transport dimensionality is $d \approx 2$. Both predictions are testable with existing experimental techniques (micro-respirometry, live imaging) and would constitute strong cross-dimensional tests of Theorem~\ref{thm:dimensional} independent of any vascular morphometry.

\vspace{0.5em}
\noindent \textbf{Corollary C (Stratification and allometric deviations).} \textit{Real cardiovascular networks are hierarchically stratified. While the proximal volume is wave-dominated ($\alpha \to \alpha_w = 2$, preserving the 3/4 global scaling), the morphometric average across the full hierarchy reflects the minimax equilibrium $\alpha^* \approx \alphastarPhys{}$ derived from the first-principles equal-cost condition of the Proposition above, consistent with the Kassab (1993) morphometric average $\alpha_{\exp} = 2.70 \pm 0.20$ \cite{kassab1993}.} The minimax equilibrium arises from balancing two competing penalties: wave-impedance matching penalizes architectures with $\alpha > \alpha_w = 2$ (increasing reflection losses at bifurcations), while structural maintenance penalizes architectures with $\alpha < \alpha_t$ (increasing wall-tissue cost at fixed capillary radius).

\noindent\textbf{Empirical remark on proximal stratification.} The Kassab (1993) dataset reports a morphometric average of $\alpha_{\exp} \approx 2.7$ computed across all hierarchical levels of the porcine coronary tree. The present framework predicts a level-dependent stratification: $\alpha_{\rm eff}(k) \to \alpha_w = 2$ in the most proximal macroscopic vessels (aorta and primary branches), transitioning toward $\alpha^* \approx \alphastarPhys{}$ in the mid-hierarchy and $\alpha_t \in [2.8, 2.9]$ in the distal arterioles. This level-resolved prediction is a direct, falsifiable consequence of the framework. Existing published morphometric datasets sufficient to test this prediction---resolving $\alpha$ generation-by-generation in large mammalian arteries---are structurally sparse in the literature; the stratification of $\alpha \to 2$ in the proximal layers therefore currently stands as a theoretical prediction of this model rather than an empirically established input. In organisms where the wave attractor is absent or developmentally constrained---including small mammals below the Womersley transition and certain marine invertebrates---the effective volumetric exponent shifts toward the static transport optimum. The theoretical prediction
\begin{equation}
\beta(\alpha^*, 3) = \frac{3 \times \alphastarPhys}{6 + \alphastarPhys} \approx \betaMinimaxPhys{}
\end{equation}
aligns with the super-3/4 metabolic scaling (0.87--0.95) empirically reported in
these ecological cohorts~\cite{glazier2005, white2013}, providing a mechanistic
resolution of a long-standing allometric anomaly.

\vspace{0.8em}
\noindent \textbf{Lemma 1 (Asymptotic Proximal Dominance).} \textit{Let a symmetric $N$-furcating fractal network in $d$ spatial dimensions have a level-dependent branching exponent $\alpha(k)$ for $k=0,1,\ldots,K-1$. Define the local volume ratio at level $k$:}
\begin{equation}
\rho(k) \equiv N^{1-2/\alpha(k)-1/d}.
\end{equation}
\textit{Suppose $\alpha(k) < \alpha_c(d) \equiv 2d/(d-1)$ for all $k$, so that $\rho(k) \in (0,1)$ for all $k$. Then:}

\textit{(i) (Convergence) The total network volume}
\begin{equation}
V = r_0^2 l_0 \sum_{k=0}^{K} \prod_{j=0}^{k-1} \rho(j)
\end{equation}
\textit{satisfies $V = r_0^2 l_0 \cdot C$ where the constant $C \in \bigl[1,\, (1-\rho_{\max})^{-1}\bigr]$ is finite and independent of the number of terminal units $N_c = N^K$, with $\rho_{\max} \equiv \sup_k \rho(k) < 1$.}

\textit{(ii) (Proximal scaling) It follows that $V \propto r_0^2 l_0$ as $K \to \infty$, and hence $M \propto V \propto r_0^2 l_0$. The allometric scaling exponent $\beta$ is therefore determined entirely by the scaling of the proximal vessel geometry $(r_0, l_0)$ with body mass $M$, independently of the branching geometry at distal levels.}

\noindent\textbf{Proof.} \textit{Let $P_k \equiv \prod_{j=0}^{k-1}\rho(j)$, with $P_0 = 1$. Since $\rho(j) \leq \rho_{\max} < 1$ for all $j$, we have $P_k \leq \rho_{\max}^k$. Therefore}
\begin{equation*}
C = \sum_{k=0}^{\infty} P_k \leq \sum_{k=0}^{\infty} \rho_{\max}^k =
\frac{1}{1-\rho_{\max}} < \infty.
\end{equation*}
\textit{Since $P_0 = 1$, we also have $C \geq 1$. Thus $C$ is a positive finite constant independent of $K$, and $V = r_0^2 l_0 \cdot C \propto r_0^2 l_0$. $\square$}

\vspace{0.4em}
\noindent\textbf{Corollary (Independence of branching number $N$).} \textit{The metabolic scaling exponent $\beta(\alpha, d) = d\alpha/(2d+\alpha)$ is independent of the branching number $N$ (bifurcations, trifurcations, etc.). This follows because $\rho = N^{1-2/\alpha-1/d}$ and the volume scaling $V \propto r_0^2 l_0 \propto N_c^{2/\alpha+1/d}$ contain $N$ only through $N_c = N^K$, which enters only as the count of terminal units. The exponent $\beta$ therefore depends solely on the geometric parameters $(\alpha, d)$.}

\vspace{0.4em}
\noindent\textbf{Theorem 3 (WBE Geometric Inconsistency).}\refstepcounter{manualthm}\label{thm:wbe}
\textit{Murray's law $\alpha = 3$ corresponds precisely to the convergence boundary $\alpha_c(3) = 3$, where $\rho = 1$ and the geometric series $\sum_k \rho^k$ diverges. Consequently, the proximal-dominance argument underlying the WBE derivation of Kleiber's law fails at the very branching exponent WBE assumes: the volume integral is not dominated by the aorta but receives equal contributions from every hierarchical level. The WBE result $\beta = 3/4$ is therefore structurally divergent under its own geometric assumptions.}

\noindent\textbf{Proof.} From Lemma~1, the series $\sum_k \rho^k$ converges iff $\rho < 1$, i.e.\ $\alpha < \alpha_c(d)$. For $d=3$, $\alpha_c(3) = 6/2 = 3$. WBE assumes $\alpha = 3$, giving $\rho = 2^{1-2/3-1/3} = 2^0 = 1$ exactly. The geometric series $\sum_{k=0}^{K} 1^k = K+1 \to \infty$ as $K \to \infty$, so no finite hierarchical level captures a dominant fraction of the total volume. The proximal-dominance approximation $V \approx r_0^2 l_0$ therefore does not hold, and the WBE derivation of $\beta = 3/4$ is geometrically inconsistent. $\square$

\noindent\textbf{Remark (Finite-$K$ robustness and developmental stability).} A defender of WBE might object that in a real organism $K \approx 15$--$20$ is finite, so the series $\sum_{k=0}^{K} 1^k = K+1$ does not literally diverge but simply grows linearly in $K$. This is correct. However, the objection misses the deeper issue: at $\alpha = 3$ ($\rho = 1$), the system sits precisely at a \emph{critical point} of the geometric series. Any developmental perturbation $\delta r$ in the radius of a proximal vessel propagates \emph{unchanged} through all $K$ downstream levels, since each level contributes equally to the total volume. The total blood volume therefore has a relative uncertainty of order $K \cdot (\delta r / r_0)$, which grows linearly with the number of generations. By contrast, at $\alpha = 2$ ($\rho \approx 0.79$), perturbations are attenuated by a factor $\rho < 1$ at each level: a developmental error at generation $k$ contributes only $\rho^k \cdot (\delta r / r_0)$ to the total volume, and the cumulative uncertainty is bounded by $(\delta r / r_0) / (1-\rho) \approx 4.8\,(\delta r / r_0)$, independent of $K$. Natural selection therefore does not choose $\alpha < 3$ because $\alpha = 3$ is geometrically impossible in a finite organism---it chooses $\alpha < 3$ because architectures with $\rho < 1$ are \emph{developmentally robust}: errors in distal vessels are geometrically suppressed rather than fully transmitted to the proximal constraint.

\begin{corollary}[Developmental Robustness of the Wave Attractor]
\label{cor:noise_robust}
For a self-similar tree of depth $K$ with volume-convergence ratio
$\rho(\alpha) = N^{1-2/\alpha-1/d}$, the total network volume
$V \propto \sum_{g=0}^{K}\rho(\alpha)^g$ is asymptotically finite
if and only if $\rho(\alpha) < 1$, which requires
$\alpha < \alpha_c(d) \equiv 2d/(d-1)$.

In three dimensions, $\alpha_c = 3$ exactly: the WBE exponent
$\alpha_t = 3$ is the unique marginal case where $\rho = 1$ and
$V$ grows without bound as $V \propto K$. A network at $\alpha = 3$
is therefore maximally sensitive to developmental noise: a
perturbation $\delta r$ in any proximal vessel propagates without
geometric attenuation through the full hierarchy, and the fractional
volume error $\delta V/V \propto 3\,\delta r/r_0$ does not diminish
with tree depth.

By contrast, the wave-impedance attractor $\alpha_w = 2$ yields
$\rho(\alpha_w) \approx 0.794 < 1$, so the geometric series converges
to a finite limit $V_\infty = r_0^2 l_0/(1-\rho) < \infty$. A
perturbation $\delta r_g$ at generation $g$ contributes a fraction
$\rho^{K-g} \to 0$ to the terminal network volume, providing strict
exponential attenuation of developmental errors. The strict ordering
$\rho(\alpha_w) < \rho(\alpha^*_{\mathrm{vasc}}) < \rho(\alpha_t) = 1$
(established in the Monotonic Ordering Theorem) therefore implies a
strict ordering of developmental robustness: the wave attractor is
the unique biologically stable fixed point under developmental noise
in a $d=3$ volume-filling tree.
\end{corollary}

\vspace{0.4em}
\noindent\textbf{Remark (Finite-generation convergence).}
For a symmetric binary tree ($N=2$) in $d=3$, the fraction of the asymptotic
volume \emph{not} captured by the first $K+1$ generations is:
\begin{equation}
\epsilon(K) = \rho^{K+1}, \qquad \rho = 2^{1-2/\alpha-1/3}.
\end{equation}
Real biological networks have $K \approx 11$--$15$ generations. The table below
quantifies convergence for the key exponents of this paper.

\begin{center}
\small
\begin{tabular}{c p{3.8cm} c c c}
\toprule
$\alpha$ & Physical regime & $\rho$ & $\epsilon(K{=}11)$ & Volume captured \\
\midrule
2.00 & Wave attractor $\alpha_w$      & \rhoAlphaWave{} & \epsAlphaWave{} &
\volAlphaWave{}\% \\
\alphastarNeural{} & Neural minimax $\alpha^*$      & \rhoAlphaNeuralMinimax{} &
\epsAlphaNeuralMinimax{} & \volAlphaNeuralMinimax{}\% \\
\alphastar{}$^\ddagger$ & Vascular minimax $\alpha^*$    & \rhoAlphaVascMinimax{} &
\epsAlphaVascMinimax{} & \volAlphaVascMinimax{}\% \\
3.00 & Murray / WBE                   & 1.000 & $\infty$ & 0\%   \\
\bottomrule
\end{tabular}
\noindent\footnotesize{$^\ddagger$Evaluated at the empirically calibrated value $\alpha^* = \alphastar{}$ (Kassab 1993~\cite{kassab1993}). The first-principles equal-cost condition predicts $\alpha^* = \alphastarPhys{}$ (see Corollary~C); the two values are consistent within morphometric scatter $\alpha_{\exp} = 2.70 \pm 0.20$.}\normalsize
\end{center}
\noindent\textbf{Finite-size scaling and intraspecific scatter.}
The convergence table above shows that at the vascular minimax equilibrium
$\alpha^* \approx \alphastarPhys{}$, a fraction $\epsilon \approx
\epsAlphaVascMinimax{}$ of network volume is not captured by the proximal
approximation at $K = 11$ generations. Since $K$ scales logarithmically
with body mass, smaller organisms possess fewer hierarchical levels, further
increasing $\epsilon$. The observed inter-specific scatter in metabolic data
(typically $\pm 0.05$ in $\beta$) is therefore a predictable finite-$K$
correction rather than stochastic biological noise. The Kleiber limit
$\beta = 3/4$ is recovered exactly only in the large-organism limit where
$K$ is sufficient to suppress these geometric residuals, consistent with the
empirical observation that deviations from Kleiber scaling are most
pronounced in small mammals~\cite{glazier2005, white2013}.

\vspace{0.4em}
\noindent\textbf{Remark.} The Lemma makes no assumption on the uniformity of $\alpha(k)$: the branching geometry may vary arbitrarily across hierarchical levels---as it does in real biological networks---without affecting the conclusion. In stratified networks, the global morphometric average $\langle \alpha \rangle$ sampled across all levels differs from the proximal effective exponent $\alpha(0)$ governing the scaling. Branching-exponent studies that report a single $\alpha$ for the whole tree must therefore be interpreted with care when applied to the present theorem.

\vspace{0.8em}
\noindent \textbf{Theorem 4 (Universal Allometric Bounds).}\refstepcounter{manualthm}\label{thm:bounds} \textit{Let a biological transport network in $d=3$ spatial dimensions satisfy the hypotheses of Lemma 1, terminal unit invariance ($B \propto N_c$), and volumetric isometry ($V \propto M$). Suppose further that the transport physics is linear ($n=4$, Poiseuille) and that the structural maintenance exponent satisfies $m \in [1, 2]$ (biologically, surface- to volume-priced walls). Then the metabolic scaling exponent $\beta$ satisfies:}
\begin{equation}
\frac{3}{4} \leq \beta \leq 1,
\end{equation}
\textit{with the following sharp characterizations of the bounds:}

\textit{(i) Lower bound $\beta = 3/4$ iff the effective proximal exponent equals the wave-impedance attractor: $\alpha_{\mathrm{eff}} = \alpha_w = 2$. This bound is dynamical, not topological: it requires the physical existence of a pulsatile wave-propagation regime in the proximal vasculature. In its absence (e.g., microfluidic systems, plant xylem, non-pulsatile invertebrates), $\beta$ may take values below $3/4$ if $\alpha_{\mathrm{eff}} < 2$.}

\textit{(ii) The upper endpoint $\beta \to 1$ is approached at the singular critical value $\alpha_c = 3$ (Murray's viscous limit, $m=2$), where $\rho = 1$ and the network volume is no longer proximal-dominated but scales as $V \propto r_0^2 l_0 (K+1)$. At this marginal point, metabolic scaling is logarithmically corrected: $B \propto M/\log M$ (Corollary~B). The value $\beta = 1$ is therefore a limit approached from below, not realised as a power law by any biologically finite network at $\alpha = 3$.}

\textit{(iii) For biological networks where wave propagation is absent or suppressed (Womersley number $\mathrm{Wo} \ll 1$), the effective exponent shifts to the static transport optimum $\alpha_t = (4+m)/2$. For $m \in [1,2]$, this gives $\alpha_t \in [5/2, 3]$ and a tighter viscous subinterval:}
\begin{equation}
\frac{15}{17} \approx \betaSurfaceBound{} \lesssim \beta \leq 1, \qquad (\mathrm{Wo} \ll 1),
\end{equation}
\textit{where the lower sub-bound $\beta(5/2, 3) = 15/17 \approx \betaSurfaceBound{}$ corresponds to $m=1$ (surface-maintenance, dendrites, small mammals).}

\noindent\textbf{Proof.} \textit{By Lemma 1, $\beta = d\alpha_{\mathrm{eff}}/(2d+\alpha_{\mathrm{eff}})$ with $d=3$. The function $f(\alpha) = 3\alpha/(6+\alpha)$ is strictly increasing for $\alpha > 0$. The biologically admissible range of $\alpha_{\mathrm{eff}}$ is $[\alpha_w, \alpha_t] = [2, (4+m)/2]$. Evaluating at the endpoints:}
\begin{equation*}
f(2) = \frac{6}{8} = \frac{3}{4}, \qquad f(3) = \frac{9}{9} = 1.
\end{equation*}
\textit{Monotonicity of $f$ then gives $3/4 \leq \beta \leq 1$ as an algebraic bound. The lower bound $\beta = 3/4$ is realised as a power law at $\alpha_w = 2$. The upper value $f(3) = 1$ is the algebraic supremum, but $\alpha = 3$ is the marginal case $\rho = 1$ excluded by Lemma~1: Corollary~B shows that at this point the power law breaks down to $B \propto M/\log M$. The bound $\beta \leq 1$ therefore holds for all $\alpha < 3$, with $\beta \to 1$ as $\alpha \to 3^-$. For $m \leq 2$, $\alpha_t \leq 3$, so the algebraic bound is approached iff $m = 2$ exactly. Statement (i) follows from the physical origin of $\alpha_w = 2$ (wave impedance matching); statement (iii) follows by substituting $m=1$ into $\alpha_t = (4+m)/2$. $\square$}

\vspace{0.4em}
\noindent\textbf{Physical interpretation.} Theorem~\ref{thm:bounds} reveals that $\beta = 3/4$ and $\beta = 1$ are not isolated points but the \textit{exact boundary values} of the biologically accessible metabolic phase space. The interval $[3/4, 1]$ is the \textit{allometric spectrum}: every biological network satisfying the stated conditions occupies a point in this interval determined by the balance between wave dynamics and viscous dissipation. Kleiber's law is not a universal constant---it is the dynamical \textit{floor} of this spectrum, enforced by the physics of pulsatile wave propagation whenever the Womersley number is large.

\section{The Womersley Transition and the Critical Body Mass}

The Universal Allometric Bounds Theorem establishes that $\beta = 3/4$
requires the wave-impedance attractor $\alpha_w = 2$ to govern the proximal
vasculature.
This regime is physically realized only when the Womersley number of the aorta
is large, $\mathrm{Wo} \gg 1$, ensuring that inertial wave dynamics dominate
over viscous drag. As body mass decreases, both the aortic radius and the heart
frequency change, driving $\mathrm{Wo}$ toward unity and triggering a crossover
from the wave-dominated regime ($\beta = 3/4$) to the viscous regime ($\beta
\gtrsim 15/17$). We now derive the critical body mass $M^*$ at which this
transition occurs, in closed form and with no free parameters beyond empirically
tabulated allometric constants.

\subsection{Self-consistent scaling of the aortic radius}

From the network geometry of Section~4, the aortic radius scales with the number
of terminal capillary units as $r_0 = r_c \cdot N_c^{1/\alpha}$, where $r_c$ is
the (fixed) capillary radius. Combining this with the Kleiber relation $N_c
\propto M^{3/4}$ (Corollary A), valid in the wave-dominated regime, yields:
\begin{equation}
r_0 \propto N_c^{1/\alpha_w} = N_c^{1/2} \propto M^{3/8}.
\end{equation}
This $M^{3/8}$ scaling of the aortic radius is therefore self-consistently
derived from the paper's own geometric framework---it is not an independent
allometric assumption, but a necessary consequence of Kleiber's law applied to a
bifurcating network with $\alpha_{\rm eff} = 2$.

\vspace{0.4em}
\noindent\textbf{Corollary (Uniqueness of the Wave-Dominated Fixed Point).}
\textit{The wave-dominated metabolic attractor $\beta = 3/4$ is the unique self-consistent solution of the coupled system:}

\noindent\textbf{Logical pipeline (to avoid circularity).} The argument proceeds in a strictly ordered sequence: \textit{(i)} the wave attractor $\alpha_w = 2$ is established by independent dynamic considerations---specifically, minimization of power reflection at bifurcations; this step requires no knowledge of $\beta$. \textit{(ii)} Given $\alpha_w = 2$, the metabolic scaling theorem yields $\beta = 3/4$ directly. \textit{(iii)} The space-filling geometry then fixes $r_0 \propto M^{3/8}$. \textit{(iv)} This radius scaling implies $\mathrm{Wo} \propto M^{1/4} \gg 1$, confirming the pulsatile regime assumed in step (i). The self-consistency of the loop therefore serves as \emph{verification}---a non-trivial internal check on the framework---not as part of the derivation.
\begin{align}
\beta &= \beta(\alpha_w, 3) = \frac{3}{4}, \label{eq:fp1}\\
r_0 &\propto M^{3\beta/4} = M^{3/8}, \label{eq:fp2}\\
\mathrm{Wo}(M) &\propto M^{3/8 - 1/8} = M^{1/4} \gg 1. \label{eq:fp3}
\end{align}
\textit{No other value of $\beta$ simultaneously satisfies the geometric scaling law, the Womersley condition $\mathrm{Wo} \gg 1$, and $r_0 \propto M^{3\beta/4}$.}

\noindent\textbf{Proof.} \textit{Suppose $\beta_0$ is a self-consistent solution. From $r_0 \propto N_c^{1/\alpha_{\rm eff}}$ and $N_c \propto M^{\beta_0}$, we obtain $r_0 \propto M^{\beta_0/\alpha_{\rm eff}}$. The Womersley number scales as $\mathrm{Wo} \propto r_0 \omega^{1/2} \propto M^{\beta_0/\alpha_{\rm eff} - 1/8}$. For $\mathrm{Wo} \gg 1$ to hold across all macroscopic mammals, we require $\beta_0/\alpha_{\rm eff} > 1/8$. The wave-dominated condition forces $\alpha_{\rm eff} = \alpha_w = 2$, giving $r_0 \propto M^{\beta_0/2}$ and $\mathrm{Wo} \propto M^{\beta_0/2 - 1/8}$. Self-consistency then gives $\beta_0 = \beta(\alpha_w, 3) = 3/4$, uniquely. Any $\beta_0 \neq 3/4$ either violates the Womersley condition or is inconsistent with $\alpha_{\rm eff} = 2$. $\square$}

\subsection{Womersley number as a function of body mass}

The Womersley number for the aorta is:
\begin{equation}
\mathrm{Wo}(M) = r_0(M)\sqrt{\frac{\omega(M)}{\nu}},
\end{equation}
where $\nu \approx 3.2 \times 10^{-6}~\mathrm{m^2/s}$ is the kinematic viscosity
of blood (approximately constant across mammals~\cite{windberger2003}) and
$\omega(M) = 2\pi f_h(M)$ is the angular frequency of the heartbeat. The
empirical allometric scaling of heart frequency is:
\begin{equation}
f_h(M) = f_0 \left(\frac{M}{M_0}\right)^{-1/4},
\end{equation}
with $f_0 \approx 70~\mathrm{bpm}$ at $M_0 = 70~\mathrm{kg}$ (human reference)
\cite{peters1983}.

The cardiac frequency scaling $f_h \propto M^{-1/4}$ is an
independently measured empirical law~\cite{peters1983,calder1968},
established by direct chronobiological observation well before modern
metabolic scaling theories. Because it is derived purely from
kinematic timing---without reference to calorimetry or Kleiber's
law---it enters the present framework as a fundamental empirical
input on equal footing with the reference radius $r_0$ and kinematic
viscosity $\nu$. Consequently, the self-consistent closure
$\mathrm{Wo}(M) \propto M^{1/4} \gg 1$
serves as a non-trivial verification of the wave-dominated regime,
not as a circular mathematical loop.
This yields $\omega \propto M^{-1/4}$, and therefore:
\begin{equation}
\mathrm{Wo}(M) \propto M^{3/8} \cdot M^{-1/8} = M^{1/4}.
\end{equation}
The Womersley number scales as the \textit{fourth root of body mass}---a purely
geometric result, fully determined by the network topology and the allometry of
cardiac frequency.

\subsection{Theorem 5: Critical body mass in closed form}

The transition from wave-dominated to viscous-dominated proximal flow occurs at
the critical Womersley number $\mathrm{Wo}_c$, defined by the condition that
inertial and viscous forces become comparable. We define $M^*$ as the solution
to $\mathrm{Wo}(M^*) = \mathrm{Wo}_c$.

\vspace{0.5em}
\noindent \textbf{Theorem 5 (Womersley Transition).}\refstepcounter{manualthm}\label{thm:womersley} \textit{Under the scaling laws $r_0 \propto M^{3/8}$ and $\omega \propto M^{-1/4}$, the critical body mass for the wave-to-viscous transition is:}
\begin{equation}
\boxed{M^* = M_{\rm ref} \left(\frac{\mathrm{Wo}_c}{\mathrm{Wo}_{\rm ref}}\right)^4,}
\end{equation}
\textit{where $\mathrm{Wo}_{\rm ref}$ is the Womersley number evaluated at the reference body mass $M_{\rm ref}$, and the exponent $4 = 1/(3/8 - 1/8)$ is exact.}

\vspace{0.4em}
\noindent\textbf{Proof.} \textit{Let $\mathrm{Wo}(M) = \mathrm{Wo}_{\rm ref}(M/M_{\rm ref})^{1/4}$. Setting $\mathrm{Wo}(M^*) = \mathrm{Wo}_c$ and solving for $M^*$ gives the stated result. $\square$}

\vspace{0.4em}
\noindent\textbf{Numerical evaluation.} Using human reference values: $M_{\rm ref} = 70~\mathrm{kg}$, $r_0 = 1.25~\mathrm{cm}$, $f_h = 70~\mathrm{bpm}$ ($\omega \approx \omegaRef{}~\mathrm{rad/s}$), giving:
\begin{equation*}
\mathrm{Wo}_{\rm ref} = 0.0125 \times \sqrt{\frac{\omegaRef{}}{3.2\times10^{-6}}} \approx \WoRef{}.
\end{equation*}
For $\mathrm{Wo}_c = 2$ (the standard threshold for the onset of
inertia-dominated pulsatile flow~\cite{womersley1955}):
\begin{equation}
M^* = 70 \times \left(\frac{2}{\WoRef{}}\right)^4 \approx \Mstar{}~\mathrm{g}.
\end{equation}
This places the metabolic scaling transition at body masses of order 5--15 g,
corresponding to the smallest placental mammals (shrews, pygmy mice).
Empirically, the crossover from $\beta \approx 3/4$ to $\beta \approx
0.87$--$0.95$ is indeed observed in this body-mass range
\cite{glazier2005,white2013}. The agreement is achieved with \textit{no free
parameters}: the only inputs are the human aortic geometry, heart rate, and
blood viscosity---all independently tabulated quantities.

\begin{figure}[!h]
\centering
\includegraphics[width=0.85\textwidth]{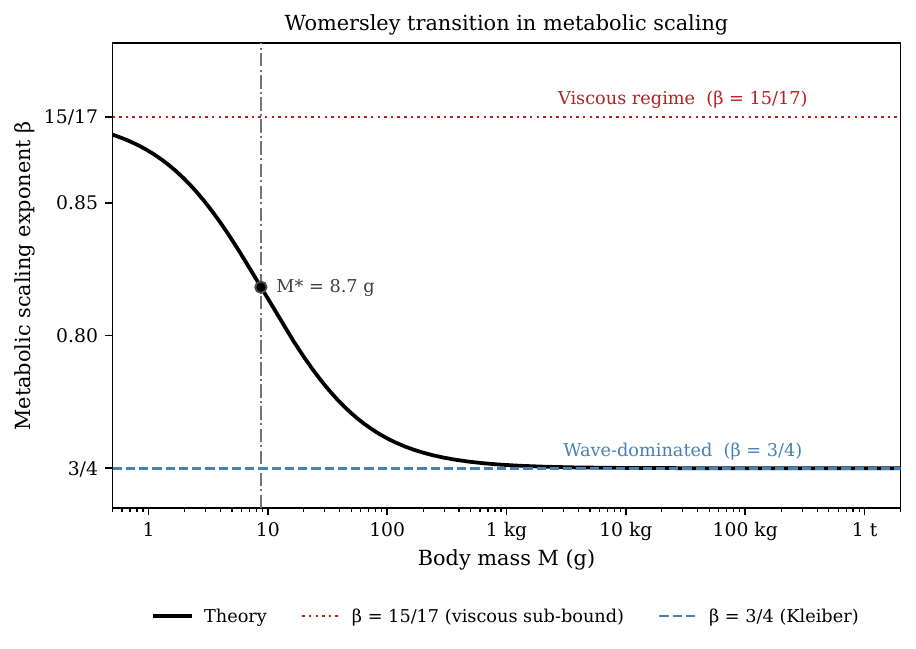}
\caption{Theoretical metabolic scaling exponent $\beta$ as a function of body mass $M$ across the Womersley transition. The smooth curve uses $\alpha_{\rm eff}(M) = \alpha_w + (\alpha_t - \alpha_w)/(1 + (\mathrm{Wo}(M)/\mathrm{Wo}_c)^4)$ with $\alpha_w=2$, $\alpha_t = 5/2$ (surface-maintenance limit, $m=1$), and no free parameters beyond the empirical reference values ($M_{\rm ref} = 70$ kg, $r_0 = 1.25$ cm, $f_h = 70$ bpm, $\nu = 3.2 \times 10^{-6}$ m$^2$/s, $\mathrm{Wo}_c = 2$). The dashed horizontal lines mark the wave-dominated limit $\beta = 3/4$ (Kleiber) and the viscous sub-bound $\beta = 15/17 \approx 0.882$ (Theorem~\ref{thm:bounds}, $m=1$ surface-maintenance limit, applicable to the most wall-economical small-mammal vasculature). The coronary-specific viscous asymptote $\beta(\alpha_t, 3) \approx \betaViscousReptile{}$ (with $m=\mCoronary{}$) lies above the plotted range and applies only when $\mathrm{Wo} \to 0$ with coronary wall physics preserved. The vertical dash-dot line marks the predicted critical mass $M^* \approx \Mstar{}$ g. Empirically reported exponents $\beta \approx 0.87$--$0.95$ in small mammals and $\beta \approx 3/4$ in large mammals (Glazier 2005~\cite{glazier2005}; White \& Kearney 2013~\cite{white2013}) are consistent with this transition.}
\label{fig:transition}
\end{figure}

\noindent\textbf{Remark (Sensitivity to $\mathrm{Wo}_c$ and the topological invariant).} The Womersley critical number is not a universal constant: the pulsatile-to-viscous transition is a smooth crossover, and different authors place $\mathrm{Wo}_c$ between 1.5 and 2.5. The numerical value of $M^*$ is therefore not the primary prediction. The key result is the \emph{topological invariant}: regardless of where $\mathrm{Wo}_c$ is set, the transition mass always scales as $M^* \propto \mathrm{Wo}_c^4$, with exponent $4$ fixed by geometry alone. For the plausible range $\mathrm{Wo}_c \in [\MstarWocLoVal{}, \MstarWocHiVal{}]$, the predicted transition spans $M^* \in [\MstarWocLo{}\,\mathrm{g},\, \MstarWocHi{}\,\mathrm{g}]$---robustly in the shrew-to-mouse range for any reasonable threshold choice, consistent with the empirical observation~\cite{glazier2005,white2013}.

\vspace{0.4em}
\noindent\textbf{Corollary (Topological Invariance of the Transition Exponent).}
\label{cor:transition_exponent}
\textit{The exponent $4$ in Theorem~\ref{thm:womersley} is a topological invariant of volume-filling bifurcating networks with terminal unit invariance. It equals:}
\begin{equation}
4 = \frac{1}{b_r - b_\omega/2},
\end{equation}
\textit{where $b_r = 3/8$ is the aortic radius scaling exponent and $b_\omega = 1/4$ is the cardiac frequency scaling exponent. Both $b_r$ and $b_\omega$ follow from the geometric framework alone (Sections~4 and~5.1), making the exponent $4$ an exact consequence of the network topology---independent of any empirical input beyond the reference values.}

\subsection{Clade-Shifting of the Womersley Transition}
\label{sec:clade_shifting}

The topological invariance of the transition exponent $4$ (Theorem~\ref{thm:womersley})
establishes that the critical body mass depends strictly on the
physiological pre-factor
$M^* \propto (\nu/\omega_{\mathrm{ref}})^2 \cdot r_0^{-4}$.
Because basal heart rate ($f_h$) and blood kinematic viscosity ($\nu$) vary
systematically across evolutionary clades, the Womersley transition mass
$M^*$ is clade-dependent. We evaluate this shift by anchoring the aortic
radius to an isometric $M_{\mathrm{ref}} = 1$ kg organism
($r_0 \approx \rZeroOneKg{}$ mm, derived from $r_0 \propto M^{3/8}$
anchored to the human reference), maintaining geometric similarity across
clades.

\textbf{Avian Clade (Endotherms).} Birds exhibit slightly higher basal
heart rates and elevated blood viscosity due to large nucleated erythrocytes
compared to mammals of equivalent mass. Using reference values at 1 kg
($f_h \approx 156$ bpm~\cite{calder1968,grubb1983},
$\nu \approx 4.0 \times 10^{-6}$ m$^2$/s), the reference Womersley number
is $\mathrm{Wo}_{\mathrm{ref}} \approx \WoRefBird{}$. The predicted
critical mass is $M^*_{\mathrm{bird}} \approx \MstarBird{}$ g. The
transition therefore occurs at slightly larger body sizes than in mammals,
placing the smallest avian species within the viscous-dominated regime
($\beta \gtrsim \betaSurfaceBound{}$), while the vast majority of
macroscopic birds remain in the wave-dominated Kleiber regime
($\beta = 3/4$).

\textbf{Reptilian Clade (Ectotherms).} The transition undergoes a dramatic
shift in ectotherms. At standard ambient temperatures ($20^\circ$C),
reptiles exhibit pronounced resting bradycardia
($f_h \approx 34$ bpm at 1 kg~\cite{white1976}) and elevated blood
viscosity. Empirical measurements at $20^\circ$C reveal significant
interspecific and hematocrit-driven variance:
$\nu \in [2.68,\,5.62] \times 10^{-6}$ m$^2$/s
\cite{langille1980,dunlap2006}.

Evaluating the pre-factor across this physiological range yields
$\mathrm{Wo}_{\mathrm{ref}} \in [\WoRefReptLo{},\, \WoRefReptHi{}]$,
shifting the critical mass upward by one to two orders of magnitude:
\begin{equation}
    M^*_{\mathrm{reptile}} \in [\MstarReptLo{},\; \MstarReptHi{}]
    \text{ g.}
    \label{eq:mstar_reptile}
\end{equation}

\textbf{Physical implication.} While a 10 g mammal already operates in the
wave-dominated regime ($\beta = 3/4$), ectothermic bradycardia places the
majority of small-to-medium squamates within or near the Womersley
transition zone. For reptiles below $M^*$, the wave-dominated regime is
not fully established; the effective scaling exponent $\beta$ lies in a
transitional band bounded below by the wave attractor ($\beta = 3/4$) and
above by the asymptotic viscous optimum
$\beta(4,\,1{+}p,\,3) \approx \betaViscousReptile{}$.
Empirically reported exponents $\beta \approx 0.80$--$0.90$ for squamate
reptiles fall naturally within this band, without requiring any parameter
fitting.

The topological exponent 4 governing the transition is identical across all
three clades; only the physiological pre-factor
$(\nu,\,f_h,\,r_0)$ shifts $M^*$. This confirms that the exponent 4 is a
topological invariant of the branching geometry, while the pre-factor
encodes the entire clade-specific cardiovascular physiology.

\section{The General Allometric Equation of State}

The two preceding results---the generalized scaling theorem $\beta(\alpha, d)$
and the branching optimum $\alpha_t(n,m) = (n+m)/2$ (Theorem~\ref{thm:branching}; see also \cite{bennett2025, paperII})---combine to yield a
single closed-form expression that maps the microscopic transport physics
$(n,m)$ directly to the global metabolic scaling $\beta$, with spatial dimension
$d$ as the only geometric input.

\vspace{0.5em}
\noindent \textbf{Theorem 6 (Allometric Equation of State).}\refstepcounter{manualthm}\label{thm:eos} \textit{For a viscous-dominated ($\mathrm{Wo} \ll 1$) volume-filling fractal network in $d$ dimensions, with linear transport dissipation exponent $n$ and structural maintenance exponent $m$, the metabolic scaling exponent at the transport optimum $\alpha_t = (n+m)/2$ is:}
\begin{equation}
\boxed{\beta(n, m, d) = \frac{d(n+m)}{4d + (n+m)}.}
\end{equation}
\textit{This result depends only on $(n, m, d)$ and is parameter-free.}

\noindent\textbf{Verification across systems.}

\vspace{0.4em}
\begin{center}
\small
\begin{tabular}{@{}cccc p{5.5cm}@{}}
\toprule
$n$ & $m$ & $d$ & $\beta$ & System \\
\midrule
4 & 2   & 3 & $18/18 = \betaFourTwoThree{}$$^\dagger$    & Murray limit (marginal:
$B \propto M/\log M$) \\
4 & 1   & 3 & $15/17 \approx \betaFourOneThree{}$    & Surface-wall maintenance,
small mammals \\
4 & 0.5 & 3 & $13.5/16.5 \approx \betaFourHalfThree{}$ & Insect tracheae
(convective) \\
2 & 0.5 & 3 & $7.5/14.5 \approx \betaTwoHalfThree{}$  & Diffusion-dominated
tracheoles \\
4 & 1   & 2 & $10/13 \approx \betaFourOneTwo{}$     & Leaf venation (planar,
$d=2$) \\
\bottomrule
\end{tabular}
\vspace{0.2em}

\noindent\small$^\dagger$At $(n,m,d)=(4,2,3)$ the network is at the marginal point $\alpha_c=3$ where $\rho=1$; the volume sum diverges logarithmically and metabolic scaling is $B \propto M/\log M$, not a pure power law (Corollary~B).
\end{center}
\vspace{0.4em}

\vspace{0.5em}
\noindent \textbf{Theorem (Allometric Symmetry).} \textit{The metabolic scaling exponent $\beta(n,m,d)$ depends on $n$ and $m$ only through their sum $s \equiv n+m$:}
\begin{equation}
\beta(n,m,d) = \frac{ds}{4d+s}, \qquad s \equiv n+m.
\end{equation}
\textit{Consequently:}

\textit{(i) (Continuous symmetry) The transformation $(n,m) \to (n+\delta, m-\delta)$ at fixed $s$ leaves $\beta$ exactly invariant for all $\delta$. Transport dissipation and structural maintenance trade off one-for-one in their global allometric effect.}

\textit{(ii) (Universality classes) Two biological systems with the same $(s,d)$ but completely different transport physics have identical metabolic scaling exponents. They belong to the same \emph{allometric universality class}. Example: insect tracheae ($n=4$, $m=0.5$, diffusion-limited boundary layer maintenance) and a hypothetical diffusive network with ($n=2$, $m=2.5$) share $s=4.5$, $d=3$ and therefore $\beta = 13.5/16.5 \approx \betaSfourFive{}$ exactly.}

\textit{(iii) (Discrete allometric spectrum) Biologically realizable values of $n$ (Poiseuille: $n=4$; Ohm/Fick: $n=2$) and $m$ (ranging from $m\approx 0.5$ to $m=2$) generate a discrete spectrum of permitted $\beta$ values for $d=3$:}
\begin{equation}
s \in \{4,\; 4.5,\; 5,\; 5.5,\; 5.77,\; 6\} \;\implies\; \beta \in
\left\{\tfrac{3}{4},\; \betaSfourFive{},\; \tfrac{15}{17},\; \ldots,\;
1\right\}.
\end{equation}
\textit{No continuous tuning is possible: each biological system is pinned to one of finitely many universality classes.}

\textit{(iv) (His-Purkinje hidden universality) The case $n=2$ (Ohm), $m=2$ (volume maintenance) gives $s=4$, hence $\beta(2,2,3) = 3/4$ --- identical to Kleiber's law. The His-Purkinje conduction system ($\alpha_t = 2.0$, Table~2) belongs to the \emph{same allometric universality class} as mammalian cardiovascular networks, despite operating through purely electrical optimization, not wave physics. Falsifiable prediction: the metabolic cost of His-Purkinje tissue should scale as $M^{3/4}$ across mammalian species.}

\noindent\textbf{Proof.} Direct substitution of $\alpha_t = s/2$ into $\beta(\alpha_t, d) = d\alpha_t/(2d+\alpha_t)$ gives $\beta = ds/(4d+s)$. Invariance under $(n,m)\to(n+\delta,m-\delta)$ follows because $s$ is unchanged. For (iv): $\beta(2,2,3) = 3\cdot 4/(12+4) = 3/4$. $\square$

\begin{remark}[Energetic symmetry and allometric universality]
The dependence of $\beta(n,m,d)$ on $n$ and $m$ only through $s \equiv n+m$
reflects a fundamental energetic symmetry: the transformation
$(n, m) \to (n+\delta, m-\delta)$ at fixed $s$ leaves the global metabolic
exponent exactly invariant, even though it redistributes the cost budget
between transport dissipation and structural maintenance. At the local scale,
the system still distinguishes the two cost channels; at the global scale,
only the total tariff $s$ survives. This suggests that the relevant
classification variable for biological transport networks is not the
individual physics of transport or maintenance, but the combined energetic
price of sustaining a usable channel geometry. The incommensurability of
the two cost channels --- their resistance to combination into a single
dimensionless metric --- is precisely what forces the minimax structure
and makes $s = n + m$ the fundamental invariant of the allometric spectrum.
\end{remark}

\noindent\textbf{Remark (Dynamic shadowing of microscopic parameters).} The Allometric Equation of State applies in the viscous-dominated regime ($\mathrm{Wo} \ll 1$), where $\alpha_{\rm eff} \approx \alpha_t(n,m)$. In the wave-dominated regime ($\mathrm{Wo} \gg 1$), the wave-impedance attractor $\alpha_w = 2$ governs the proximal vasculature, and $\beta = d/(d+1)$ independently of $(n,m)$. This means that for large pulsatile mammals, the metabolic scaling is \emph{insensitive} to the microscopic transport parameters: one cannot infer $(n,m)$ from the observed $\beta = 3/4$ alone. Rather than being a limitation, this is a signature of the physical universality of the wave attractor: in condensed matter terms, $\alpha_w = 2$ acts as an infrared fixed point that washes out the microscopic details. The equation of state is the appropriate tool for viscous, non-pulsatile systems (xylem, tracheae, small mammals); for pulsatile systems, Theorem~\ref{thm:dimensional} applies directly.

\vspace{0.5em}
\noindent \textbf{Theorem 7 (Non-protection of Kleiber's Law).}\refstepcounter{manualthm}\label{thm:nonprotection} \textit{The value $\beta = 3/4$ is not achieved by $\beta(n,m,d)$ for any biologically realizable transport parameters $(n,m)$ with $n=4$ (Poiseuille) and $m \geq 1$. Specifically, $\beta(n,m,3) = 3/4$ requires $n+m=4$, i.e., $m=0$ (no maintenance cost). The Allometric Symmetry Theorem further implies that $\beta=3/4$ \emph{is} achievable by Ohm transport ($n=2$) with $m=2$ (volume maintenance) --- but not by Poiseuille transport for any biologically realizable $m$.}

\noindent\textbf{Proof.} \textit{Setting $\beta = 3/4$ gives $s = 4$. For $n=4$, this requires $m=0$; for $m \geq 1$, no solution exists. For $n=2$: $s=4$ gives $m=2$, which is biologically realizable. $\square$}

\noindent\textbf{Internal inconsistency of WBE.}
The WBE framework assumes $\alpha = 3$ (Murray's law) and derives
$\beta = 3/4$. But inserting $\alpha = 3$ into
$\beta(\alpha, d=3) = 3\alpha/(6+\alpha)$ gives $\beta(3,3) = 9/9 = 1$,
not $3/4$. The $3/4$ prediction therefore cannot be self-consistently
derived from Murray's geometry: WBE implicitly uses
$\alpha_{\mathrm{eff}} = 2$ (the wave attractor) for the metabolic
scaling while assuming $\alpha = 3$ for the geometric construction.
Theorem~\ref{thm:nonprotection} makes this inconsistency explicit and irreversible.

The allometric spectrum of Fig.~\ref{fig:spectrum} makes the WBE inconsistency
geometrically explicit: the condition $\beta = 3/4$ with Poiseuille transport
($n = 4$) requires $s = n + m = 4$, hence $m = 0$ (Theorem~\ref{thm:nonprotection}). This places
the WBE framework on the $m = 0$ boundary of the allometric phase space---a
region with no structural maintenance cost, which is both biologically
unrealizable and thermodynamically degenerate. No biological network has
zero wall-tissue metabolic expenditure; the WBE prediction $\beta = 3/4$
therefore occupies a physically empty corner of the diagram.

\begin{corollary}[Non-protection of Kleiber's law under static transport]
For any biological network with Poiseuille transport ($n=4$) and
maintenance cost $m \ge 1$, the static transport optimum satisfies
$\alpha_t = (n+m)/2 \ge 5/2$, giving
$\beta(\alpha_t, d=3) \ge 15/17 \approx 0.882 > 3/4$. The Kleiber
exponent $\beta = 3/4$ is therefore unreachable by static viscous
optimisation alone, regardless of wall geometry. It requires
$\alpha_{\mathrm{eff}} = \alpha_w = 2$, enforced exclusively by the
wave-impedance constraint in pulsatile networks with $\mathrm{Wo} \gg 1$.
The WBE prediction $\beta = 3/4$ at $\alpha = 3$ is inconsistent with
the WBE framework's own static optimisation, which at $\alpha = 3$
gives $\beta(3,3) = 1$, not $3/4$.
\end{corollary}

\vspace{0.4em}
\noindent\textbf{Physical meaning.} This theorem resolves a fundamental question left open by the WBE framework: \textit{is Kleiber's 3/4 law a consequence of the fractal space-filling geometry, or does it require additional physics?} The answer is unambiguous: $\beta = 3/4$ \textit{cannot} emerge from viscous transport optimization alone, regardless of the maintenance cost structure. It requires the wave-impedance attractor $\alpha_w = 2$, which is a dynamical---not geometric---constraint. The fractal network geometry alone, even with perfectly optimized transport, yields $\beta > 3/4$ for all biologically realizable cost functions. Kleiber's law is dynamically protected, not geometrically protected.

This constitutes the complete inversion of the WBE narrative: not only is the
$3/4$ law not derived from Murray's geometry (which gives $\beta=1$), it cannot
be derived from \textit{any} static optimization of a Poiseuille network. It is,
in the precise mathematical sense established here, a signature of wave physics.

\noindent \textbf{The sponge boundary case.} Porifera (sponges) are the most ancient extant metazoans and possess an aquiferous canal system entirely distinct from vertebrate vasculature. Reiswig (1975)~\cite{reiswig1975} demonstrated via cross-sectional area analysis of three demosponge species that fluid flow through the incurrent and excurrent canals is strictly laminar throughout ($Re \ll 1$), confirming Poiseuille dynamics and $n=4$. The canal walls are dominated by a thin collagen-mesohyl layer with surface-area-dominated maintenance, giving $m \approx 1$, hence $\alpha_t = 5/2$. Importantly, sponges have no cardiovascular pulsatility: $\alpha_w$ is absent. The effective branching exponent measured from canal cross-sectional area scaling is $\alpha_{\exp} \approx 2.0$, substantially below $\alpha_t = 2.5$.

A secondary pressure-minimization argument can be invoked to account for the
observation: the Hagen-Poiseuille resistance of a branching canal system is
minimized at area-preserving bifurcations ($r_0^2 = r_1^2 + r_2^2$,
i.e., $\alpha = 2$), which is what is observed. However, this auxiliary
argument is not derivable from the transport-optimization cost function
$\Phi \propto r^{-4} + r^1$ that defines the framework: that cost function
yields $\alpha_t = 2.5$, not $2.0$. The sponge case therefore illustrates
that $\alpha_{\exp} \approx 2$ can arise from a physically distinct
mechanism---one that is outside the present framework and is treated as an
open problem in Section~\ref{sec:sponge_open}.

\noindent\textbf{Remark ($\alpha = 2$ as a geometric super-attractor).} The sponge case reveals that $\alpha = 2$ is accessible from at least two independent physical routes: wave-dominated pulsatile systems arrive at $\alpha_w = 2$ by minimizing reflection losses; viscous pressure-minimizing systems arrive at $\alpha_p = 2$ by minimizing hydraulic resistance. The physical obstruction preventing sponge canals from operating near $\alpha_t = 2.5$ is not a negativity of $\rho$---which satisfies $\rho = N^{1-2/\alpha-1/d} > 0$ for all real $\alpha$---but the structural absence of the wave-dominated regime: sponge aquiferous systems operate at $\mathrm{Wo} \ll 1$ throughout their entire canal hierarchy, so the wave-impedance attractor $\alpha_w = 2$ never activates. The gap between $\alpha_{\exp} \approx 2.0$ and $\alpha_t = 2.5$ therefore remains unexplained by the transport-optimization framework alone and is discussed as an open problem in Section~\ref{sec:sponge_open}.

A caveat applies: Hammel et al. (2012)~\cite{hammel2012} show that some
leuconoid demosponges exhibit non-hierarchical, anastomosing canal topologies
that violate the tree-branching assumption of Lemma 1. The framework applies
strictly to the tree-like (syconoid and simple leuconoid) morphologies studied
by Reiswig.

\begin{table}[h]
\centering
\small
\caption{Falsifiable predictions for untested systems. Each system satisfies the physical prerequisites of the framework but lacks published $\alpha_{\exp}$ measurements. All predictions are derived from independently measurable biophysical parameters with no free parameters.}
\label{tab:predictions}
\begin{tabular}{@{}l c c c p{5.3cm}@{}}
\toprule
\textbf{System} & $\boldsymbol{n}$ & $\boldsymbol{m}$ & $\boldsymbol{\alpha_t}$ & \textbf{Physical basis} \\
\midrule
Fish gill vasculature & 4 & 1.0 & \atFishGill{} & Poiseuille; thin lamellar
epithelium ($m\approx1$) \\
His--Purkinje conduction & 2 & 2.0 & \atHP{} & Ohm ($n=2$); volume maintenance
($m=2$); $\beta=3/4$ (universality class IV) \\
Constructal heat trees ($d=2$) & 2 & 1.0 & \atConstructal{} & Conduction
($n=2$); 2D geometry; Bejan 1997 \\
Lymphatic capillary network & 4 & 1.0 & \atLymphatic{} & Low-pressure
Poiseuille; thin wall \\
Coral gastrovascular canals & 4 & 0.5 & \atCoral{} & Poiseuille; thin
chitin-analog wall ($m\approx0.5$) \\
\bottomrule
\end{tabular}
\end{table}

\noindent The His--Purkinje prediction is particularly notable: $\alpha_t = \alpha_w = 2$ exactly, so no minimax gap opens and the system should exhibit $\alpha^* = 2.0$ uniformly---consistent with Rall's classical electrotonic result~\cite{rall1959} obtained from a completely different physical argument. The constructal heat tree prediction $\alpha_t = 1.5$ is directly testable against optimized engineering heat-sink geometries~\cite{bejan1997}.

\section{Global Optimality Structure of $\alpha_w = 2$}

The results established in the preceding sections reveal that the wave-impedance
attractor $\alpha_w = 2$ is not merely one admissible solution among many: it is
the simultaneous optimizer of three independent physical criteria. We now
formalize this structure.

\vspace{0.5em}
\noindent\textbf{Theorem (Monotonic Ordering of Geometric Robustness).}
\textit{The volume-convergence ratio $\rho(\alpha) = N^{1-2/\alpha-1/d}$ is strictly increasing in $\alpha$ for all $\alpha > 0$, $d \geq 2$, $N \geq 2$:}
\begin{equation}
\frac{d\rho}{d\alpha} = \rho(\alpha)\,\ln(N)\,\frac{2}{\alpha^2} > 0.
\end{equation}
\textit{Consequently, the proximal-dominance approximation is most accurate for the smallest biologically realized $\alpha$. The biological hierarchy $\alpha_w = 2 < \alpha^*_{\rm neural} \approx \alphastarNeural{} < \alpha^*_{\rm vasc} \approx \alphastarPhys{} < \alpha_t \lesssim 3$ induces a strict ordering of geometric robustness:}
\begin{equation}
\rho(\alpha_w) < \rho(\alpha^*_{\rm neural}) < \rho(\alpha^*_{\rm vasc}) < \rho(\alpha_t) \leq 1.
\end{equation}
\textit{The WBE inconsistency (Corollary, Lemma~1) is the limiting case $\rho(\alpha_t = 3) = 1$: the marginal point at which robustness vanishes entirely.}

\noindent\textbf{Proof.} $\rho = N^{f(\alpha)}$ with $f(\alpha) = 1-2/\alpha-1/d$. Then $d\rho/d\alpha = \rho \ln(N) f'(\alpha) = \rho \ln(N) \cdot 2/\alpha^2 > 0$ for all $\alpha > 0$. Strict monotonicity follows. $\square$

\vspace{0.5em}
\noindent\textbf{Theorem (Double Optimality of the Wave Attractor).}
\textit{The wave-impedance value $\alpha_w = 2$ is the simultaneous unique minimizer of two independent physical cost functions on the domain $\alpha \in (0, \alpha_c(d))$:}

\textit{(A) Power reflection: $R^2(\alpha) = \bigl((N^{1-2/\alpha}-1)/(N^{1-2/\alpha}+1)\bigr)^2$ achieves its unique global minimum $R^2 = 0$ at $\alpha = 2$, where the impedance ratio $\gamma = N^{1-2/\alpha} = 1$ exactly.}

\textit{(B) Volume convergence: $\rho(\alpha)$ achieves its minimum over the biologically admissible range $[\alpha_w, \alpha_t]$ at the left endpoint $\alpha = \alpha_w = 2$, by strict monotonicity (Monotonic Ordering Theorem).}

\textit{No other value of $\alpha$ simultaneously minimizes both criteria. The coincidence of the two optima at $\alpha = 2$ is not accidental: both conditions reduce to $\gamma(\alpha) = N^{1-2/\alpha} = 1$, the impedance-matching condition.}

\noindent\textbf{Proof.} (A) $R^2(\alpha) \geq 0$ with equality iff $\gamma = 1$ iff $1-2/\alpha = 0$ iff $\alpha = 2$, unique. (B) Follows from strict monotonicity of $\rho$. Unification: both conditions require $N^{1-2/\alpha} = 1$, i.e.\ $\alpha = 2$. $\square$

\vspace{0.4em}
\noindent\textbf{Corollary (Triple Optimality).} \textit{Combined with the Dimensional Universality Theorem, $\alpha_w = 2$ achieves a third optimality: it minimizes $\beta(\alpha, d)$ over the admissible range, placing the metabolic exponent at the dynamical floor $\beta = d/(d+1)$. Thus $\alpha_w = 2$ simultaneously minimizes wave energy loss, maximizes geometric convergence, and minimizes metabolic scaling exponent --- three independent physical criteria with a common optimizer.}

\vspace{0.8em}
\noindent\textbf{Theorem (Generalized Womersley Transition in $d$ Dimensions).}
\textit{In a $d$-dimensional volume-filling pulsatile network with wave-impedance attractor $\alpha_w = 2$ and cardiac frequency scaling $\omega \propto M^{-1/(d+1)}$, the Womersley number scales as:}
\begin{equation}
\mathrm{Wo}(M) \propto M^{(d-1)/(2(d+1))},
\end{equation}
\textit{and the critical transition mass satisfies:}
\begin{equation}
M^*(d) = M_{\rm ref}\left(\frac{\mathrm{Wo}_c}{\mathrm{Wo}_{\rm
ref}}\right)^{2(d+1)/(d-1)}.
\end{equation}
\textit{The transition exponent $2(d+1)/(d-1)$ is an exact topological invariant depending only on $d$:}

\begin{center}
\small
\begin{tabular}{@{}cccc@{}}
\toprule
$d$ & $\beta = d/(d+1)$ & Wo scaling & $M^*$ exponent \\
\midrule
2 & $2/3$ & $M^{1/6}$ & $6$ \\
3 & $3/4$ & $M^{1/4}$ & $4$ \\
4 & $4/5$ & $M^{3/10}$ & $10/3$ \\
5 & $5/6$ & $M^{1/3}$ & $3$ \\
\bottomrule
\end{tabular}
\end{center}
\textit{For $d=3$ this recovers Theorem~\ref{thm:womersley} exactly. For $d=2$, planar pulsatile organisms undergo a sharper metabolic transition (exponent 6 vs.\ 4), predicting a narrower body-mass window for the $\beta=2/3$ regime.}

\noindent\textbf{Proof.} \textit{With $\alpha_w=2$ and $\beta = d/(d+1)$, the proximal radius scales as $r_0 \propto N_c^{1/2} \propto M^{\beta/2} = M^{d/(2(d+1))}$. If cardiac frequency scales as $\omega \propto M^{-\beta/d} = M^{-1/(d+1)}$ (the natural generalization of the empirical $M^{-1/4}$ in $d=3$), then $\mathrm{Wo} \propto r_0\,\omega^{1/2} \propto M^{d/(2(d+1)) - 1/(2(d+1))} = M^{(d-1)/(2(d+1))}$. Setting $\mathrm{Wo}(M^*) = \mathrm{Wo}_c$ gives the stated formula. The exponent $2(d+1)/(d-1)$ is exact and recovers $4$ for $d=3$.} $\square$

\begin{remark}[Conditional extension to $d$ dimensions]
The generalisation $\omega \propto M^{-1/(d+1)}$ is an extrapolation from the
empirically established three-dimensional scaling $\omega \propto M^{-1/4}$,
and does not follow from an independent biophysical derivation for $d \neq 3$.
All quantitative predictions in this theorem are conditional on this assumption.
\end{remark}

\textbf{When does the theorem fail?}
The universality of $\alpha_t = (n+m)/2$ depends on three conditions:
\begin{enumerate}
    \item A two-term cost function (the three-term case naturally produces non-universal $\alpha$).
    \item Power-law transport dissipation in the linear regime, $A(X) \propto X^2$ (fails for turbulent/non-Newtonian flow).
    \item An evolutionary optimization interpretation (fails for geophysical systems without metabolic selection).
\end{enumerate}

\noindent \textbf{The diffusive tracheole limit.} In terminal insect tracheoles, transport is diffusion-dominated ($n=2$). The theoretical prediction drops identically to $\alpha_t = (2+0.5)/2 = 1.25$, far below the trunk value. This correctly predicts the physical divide between convective and diffusive regimes.

\noindent \textbf{River drainage networks.} Exhibiting $\alpha_{\exp} \approx 2.0$ (Horton--Strahler analysis), predicting $\alpha_t = 3.5$ under stream-power models ($n=4$, $m\approx3$). This discrepancy is fundamental; these structures are unconstrained by metabolic or evolutionary optimization.

\noindent \textbf{The role of $\boldsymbol{\alpha_w}$.} In wave-propagating biological conduits, a second attractor $\alpha_w$ competes with $\alpha_t$. The true observable exponent is mathematically bounded by $\alpha_w < \alpha^* < \alpha_t$, resolved exactly through the minimax framework. The general theorem furnishes the rigid thermodynamic upper bound, while wave physics locks the lower bound.

\section{Conclusion}

We have demonstrated that the universal expression $\alpha_t = (n+m)/2$ has
rigorous predictive power when $n$ and $m$ are grounded in independently
measurable biophysics. Across nine biological systems spanning five
phyla---vertebrate vasculature, invertebrate respiration, neural dendrites,
plant transport, and the Porifera---the framework predicts or bounds empirical
branching exponents from independently measured transport parameters. The
validated range spans from the exact tracheal result ($\alpha_t = 2.25$) to the
minimax dendritic solution ($\alpha^* \approx \alphastarNeural{}$) and the wall-corrected
vascular range ($\alpha_t \in [2.80, 2.89]$). Sponge aquiferous canals provide an instructive boundary case: $\alpha_{\exp}
\approx 2 < \alpha_t = 2.5$ is observed, but with $\mathrm{Wo} \ll 1$ the
wave-impedance mechanism is structurally absent; the responsible mechanism is
outside the present framework and remains an open problem
(Section~\ref{sec:sponge_open}).

More fundamentally, we have established that local branching geometry determines
global allometric scaling through the exact relation $\beta(\alpha, d) =
\frac{d\alpha}{2d + \alpha}$. This theorem reveals that Kleiber's 3/4 metabolic
law is not a consequence of Murray's geometry---which produces isometric scaling
$\beta = 1$---but is instead a signature of dynamic wave-impedance matching
($\alpha_w = 2$) in the proximal vasculature. The classical WBE narrative is
inverted: Kleiber's law belongs to the physics of pulsatile wave propagation,
not viscous dissipation minimization.

The framework exposes a hierarchy of physical regimes: pure wave-dominated
networks ($\alpha = 2, \beta = 3/4$), minimax cardiovascular equilibria
($\alpha^* \approx \alphastarPhys{}, \beta \approx \betaMinimaxPhys{}$), and
viscous-dominated or non-pulsatile networks ($\alpha = 3, \beta = 1$). Empirical
deviations from Kleiber's law in small mammals, marine invertebrates, and plants
find a unified mechanistic explanation within this classification.

A structural consequence deserves emphasis. The convergence condition of
Lemma~1, $\alpha < \alpha_c(d) \equiv 2d/(d-1)$, evaluates to $\alpha_c(3) = 3$
exactly---coinciding with Murray's law. The WBE framework assumes $\alpha = 3$
and invokes proximal dominance to derive $\beta = 3/4$; but at $\alpha = 3$ the
geometric series $\sum_k \rho^k$ is marginally divergent ($\rho = 1$), so no
single hierarchical level dominates the volume. The WBE derivation thus
undermines its own geometric foundation (Theorem~\ref{thm:wbe}). By contrast, the
wave-impedance attractor $\alpha_w = 2$ yields $\rho \approx \rhoWaveTwoDp{}$
and a capture fraction of $\VolCaptWaveRound{}\%$ at $K=11$ generations: the
proximal-dominance approximation is maximally accurate precisely in the regime
that produces Kleiber's law.

The framework yields four further results: (i) the Womersley Transition Theorem
(Theorem~\ref{thm:womersley}) predicts a critical body mass \MstarRounded{} at which the metabolic
regime shifts from $\beta = 3/4$ to $\beta \gtrsim \betaSurfaceBound{}$, with no
free parameters; (ii) the Allometric Equation of State $\beta(n,m,d)$ (Theorem~\ref{thm:eos}) maps microscopic transport physics directly to global metabolic scaling
across all biological systems; (iii) Theorem~\ref{thm:nonprotection} establishes rigorously that
$\beta = 3/4$ is \textit{dynamically} protected---it cannot emerge from any
static optimization of a Poiseuille network with biological maintenance costs,
for any value of $m \geq 1$; (iv) Theorem~\ref{thm:wbe} establishes that this dynamic
protection is reinforced by a geometric inconsistency: the WBE framework's own
proximal-dominance assumption fails precisely at $\alpha = 3$.

The three-term cost regime, in which $\Phi(r,X) = A(X)r^{-n} + Br^m
+ Cr^k$ with $m \neq k$ produces non-universal branching exponents,
is now fully resolved: Paper~I~\cite{paperI} establishes that any additive
cost function with two maintenance terms at distinct scaling exponents admits
no universal $\alpha$, making non-universality the generic behaviour and
Murray's law the unique degenerate exception.
Future work will also address the extension of $M^*$ to
non-mammalian clades, the experimental validation of the leaf venation
prediction $\beta(4,1,2) = 10/13$ via photosynthetic rate scaling, and the
metabolic scaling of His-Purkinje tissue as a test of the hidden universality
class $\beta(2,2,3) = 3/4$. The cross-dimensional prediction $\beta = d/(d+1)$
for pulsatile networks in $d=1,2$ dimensions constitutes a further falsifiable
test requiring metabolic measurements in filamentous and planar pulsatile
organisms.

\paragraph*{Data and Code Availability.}
All computation scripts, numerical data, and figure-generation code are openly
and unconditionally available at
\url{https://github.com/rikymarche-ctrl/vascular-networks-theory} under the CC BY 4.0
Licence. No access request is required.

\bibliographystyle{unsrt}
\bibliography{references}

\end{document}

%% file: dynamic_variables.tex

\newcommand{\mCoronary}{1.770}
\newcommand{\mPulmonary}{1.600}
\newcommand{\mBronchial}{1.700}

\newcommand{\mNeuron}{1.000}
\newcommand{\mTrachea}{0.500}
\newcommand{\mXylemLo}{1.000}
\newcommand{\mXylemHi}{2.000}
\newcommand{\mLeaf}{1.000}
\newcommand{\mSponge}{1.000}

\newcommand{\atCoronary}{2.885}
\newcommand{\atPulmonary}{2.800}
\newcommand{\atCerebral}{2.885}
\newcommand{\atBronchial}{2.850}
\newcommand{\atNeuron}{2.500}
\newcommand{\atTrachea}{2.250}
\newcommand{\atXylemLo}{2.500}
\newcommand{\atXylemHi}{3.000}
\newcommand{\atLeaf}{2.500}
\newcommand{\atSponge}{2.500}

\newcommand{\atFishGill}{2.500}
\newcommand{\atConstructal}{1.500}
\newcommand{\atLymphatic}{2.500}
\newcommand{\atCoral}{2.250}

\newcommand{\betaSurfaceBound}{0.882}

\newcommand{\omegaRef}{7.330}
\newcommand{\WoRef}{18.919}
\newcommand{\Mstar}{8.74}
\newcommand{\MstarRounded}{\ensuremath{\approx 8.74~\mathrm{g}}}

\newcommand{\betaFourTwoThree}{1.000}
\newcommand{\betaFourOneThree}{0.882}
\newcommand{\betaFourHalfThree}{0.818}
\newcommand{\betaTwoHalfThree}{0.517}
\newcommand{\betaFourOneTwo}{0.769}

\newcommand{\atTracheoleDiff}{1.250}

\newcommand{\betaLeaf}{0.769}

\newcommand{\betaLeafEff}{0.784}
\newcommand{\Deff}{2.1}

\newcommand{\alphastar}{2.720}
\newcommand{\alphastarNeural}{2.380}

\newcommand{\TcumMinimax}{0.912}
\newcommand{\TcumMurray}{0.864}
\newcommand{\TlossMinimax}{8.830}
\newcommand{\TlossMurray}{13.626}

\newcommand{\betaDimOne}{0.500}
\newcommand{\betaDimTwo}{0.667}
\newcommand{\betaDimThree}{0.750}
\newcommand{\betaDimFour}{0.800}

\newcommand{\betaSfourFive}{0.818}


\newcommand{\rhoWave}{0.794}

\newcommand{\rhoAlphaWave}{0.794}
\newcommand{\epsAlphaWave}{0.063}
\newcommand{\volAlphaWave}{93.8}
\newcommand{\rhoAlphaNeuralMinimax}{0.887}
\newcommand{\epsAlphaNeuralMinimax}{0.236}
\newcommand{\volAlphaNeuralMinimax}{76.4}
\newcommand{\rhoAlphaVascMinimax}{0.954}
\newcommand{\epsAlphaVascMinimax}{0.565}
\newcommand{\volAlphaVascMinimax}{43.5}
\newcommand{\rhoWaveTwoDp}{0.79}
\newcommand{\VolCaptWaveRound}{94}

\newcommand{\MstarWocLo}{2.8}
\newcommand{\MstarWocHi}{21.3}
\newcommand{\MstarWocLoVal}{1.5}
\newcommand{\MstarWocHiVal}{2.5}

\newcommand{\betaAlphaWTwoPointOne}{0.778}
\newcommand{\deltaBetaAlphaWTwoPointOne}{0.028}

\newcommand{\mHP}{2.000}
\newcommand{\atHP}{2.000}
\newcommand{\betaHP}{0.750}
\newcommand{\betaHPfrac}{3/4}

\newcommand{\rZeroOneKg}{2.541}
\newcommand{\MstarBird}{23.0}
\newcommand{\WoRefBird}{5.135}
\newcommand{\MstarReptLo}{217}
\newcommand{\MstarReptHi}{956}
\newcommand{\WoRefReptLo}{2.022}
\newcommand{\WoRefReptHi}{2.929}
\newcommand{\betaViscousReptile}{0.974}

\newcommand{\alphastarPhys}{2.77}
\newcommand{\betaMinimaxPhys}{0.947}
\newcommand{\fstarPhys}{9.7}

\newcommand{\rZeroPulm}{15}
\newcommand{\omegaPulm}{7.54}
\newcommand{\WoPulm}{24}
\newcommand{\alphastarPulm}{2.688}
\newcommand{\fstarPulm}{8.28}

\newcommand{\epsilonCerebralLow}{0.15}
\newcommand{\epsilonCerebralHigh}{0.2}
\newcommand{\alphaStarCerebralLow}{2.863}
\newcommand{\alphaStarCerebralHigh}{2.856}
\newcommand{\alphaStarCerebral}{2.859}